\documentclass[runningheads,a4paper,envcountsame]{llncs}
\usepackage{graphicx}
\usepackage{amsmath, amssymb, color}
\usepackage{caption}
\usepackage{subfig} 

\usepackage{epsfig}
\usepackage{graphics}
\usepackage{array}
\usepackage{amsmath}
\usepackage{amsfonts}
\usepackage{amssymb}
\usepackage[algo2e,ruled,vlined,linesnumbered]{algorithm2e}

\usepackage{setspace}
\usepackage{fullpage} 

\usepackage{color}

\newtheorem{observation}[theorem]{Observation}

\newenvironment{proposition-a}[1]{\noindent {\bf Proposition~#1.~}\em }{\smallskip}
\newenvironment{lemma-a}[1]{\noindent {\bf Lemma~#1.~}\em }{\smallskip}
\newenvironment{theorem-a}[1]{\noindent {\bf Theorem~#1.~}\em }{\smallskip}

\DeclareMathOperator{\bound}{bound}
\DeclareMathOperator{\Sup}{Support}
\DeclareMathOperator{\cost}{cost}
\DeclareMathOperator{\size}{size}

\DeclareMathOperator{\avgstr}{AvgStr}
\begin{document}
\title{On Minimum Average Stretch Spanning Trees in Polygonal 2-trees \thanks{Supported by the Indo-Max Planck Centre for Computer Science Programme in the area of \emph{Algebraic and Parameterized Complexity} for the year 2012 - 2013} \thanks{The preliminary version of this work is appeared in Eighth International Workshop on Algorithms and Computation (WALCOM) 2014}}
\author{N.S.Narayanaswamy \and G.Ramakrishna}
\institute{Department of Computer Science and Engineering\\ Indian Institute of Technology Madras, India.
\\ \email{\{swamy,grama\}@cse.iitm.ac.in}
}

\titlerunning{\textsc{Mast} in polygonal 2-trees}

\authorrunning{Narayanaswamy, Ramakrishna}

\maketitle

\sloppy{ 
\begin{abstract}
 A spanning tree of an unweighted graph is a \emph{minimum average stretch spanning tree} if it minimizes the ratio 
of sum of the distances in the tree between the end vertices of the graph edges 
and the number of graph edges.  We consider
the problem of computing a minimum average stretch spanning tree in polygonal 2-trees, a super class of 2-connected outerplanar graphs.
For a polygonal 2-tree on $n$ vertices, we present an algorithm to compute a minimum average stretch spanning tree in $O(n \log n)$ time. This algorithm also finds a minimum fundamental cycle basis in polygonal 2-trees.
\end{abstract}
}

\sloppy{
\section{Introduction}
Average stretch is a parameter used to measure the quality of a spanning tree in terms of distance preservation, and finding a spanning tree with minimum average stretch is a classical problem in network design.
Let $G=(V(G),E(G))$ be an unweighted graph and $T$ be a spanning tree of $G$. For an  edge $(u,v) \in E(G)$, $d_{T}(u,v)$ denotes the distance between $u$ and $v$ in $T$.
The average stretch of $T$ is defined as 

\begin{equation}
\label{EquationAveragStretch}
\avgstr(T) = \frac{1}{|E(G)|}\sum_{ (u,v) \in E(G) } d_{T}(u,v) 
\end{equation}

\noindent
 A \emph{minimum average stretch spanning tree} of $G$ is a spanning tree that minimizes the average stretch. Given an unweighted graph $G$, the minimum average stretch spanning tree (\textsc{Mast}) problem is to find a minimum  average stretch spanning tree of $G$.  Due to the unified notation for tree spanners, the \textsc{Mast} problem is equivalent to 
the problem, \textsc{Mfcb}, of finding a minimum fundamental cycle basis in unweighted graphs \cite{LiebchenZooOfTreeSpanners}. 
Minimum average stretch spanning trees are used to solve symmetric diagonally dominant linear systems \cite{LiebchenZooOfTreeSpanners}.
Further, minimum fundamental cycle bases have various applications including determining the isomorphism of graphs, frequency analysis of computer programs, and generation of minimal perfect hash functions
 (See \cite{Deo_MFCB,GalbiatiApprox} and the references there in]). 
Due to these vast applications, finding a  minimum average stretch spanning tree is useful in theory and practice.
 The \textsc{Mast} problem was studied in a graph theoretic game in the context of the $k$-server problem by Alon et al. \cite{Alon95graphTheoretic}. 
The   \textsc{Mfcb} problem was introduced by Hubika and Syslo in 1975 \cite{hubikaSyslo}.
The \textsc{Mfcb} problem was proved to be NP-hard by Deo et al. \cite{Deo_MFCB} and APX-hard by Galbiati et al. \cite{GalbiatiApprox}.
  Another closely related problem is the problem of probabilistically embedding a graph into its spanning trees.  A graph $G$ is said to be \emph{probabilistically embedded} into its spanning trees with \emph{distortion} $t$, if there is a probability distribution $D$ of spanning trees of $G$, such that for any two vertices the expected stretch of the spanning trees in $D$ is at most $t$. The problem of probabilistically embedding a graph into its spanning trees with low distortion has interesting connections with low average stretch spanning trees.
  
 \noindent
In the  literature,  spanning trees with low average stretch has received significant attention in special graph classes such as $k$-outerplanar graphs and series-parallel graphs.  In case of planar graphs, Kavitha et al. remarked that the complexity of \textsc{Mfcb} is unknown and there is no $O(\log n)$ approximation algorithm \cite{KavithaSurvey}.  
For $k$-outerplanar graphs, the technique of peeling-an-onion decomposition is employed to obtain 
a spanning tree whose average stretch is at most $c^{k}$, where $c$ is a constant \cite{EmekOuterplanar}.
In case of series-parallel graphs, a spanning tree with average stretch at most $O( \log n)$ can be obtained in polynomial time (See Section 5 in \cite{EmekSeriesParallelPE}).  The bounds on the size of a minimum fundamental cycle basis is studied in graph classes such as planar, outerplanar and grid graphs \cite{KavithaSurvey}. The study of probabilistic embeddings of graphs is discussed in \cite{EmekOuterplanar,EmekSeriesParallelPE}.   
To the best of our knowledge, there is no published work to compute a minimum average stretch spanning tree and minimum fundamental cycle basis in any subclass of planar graphs.

\noindent
We consider polygonal 2-trees in this work, which are also referred to as polygonal-trees.  They have a rich structure that make them very natural models for  biochemical compounds, and provide an appealing framework for solving associated enumeration problems.

\begin{definition}[\cite{KohSPgraphs1996}]
\label{DefinitionPolygonal2Tree}
A cycle is a polygonal 2-tree. For a polygonal 2-tree $G$ such that $(u,v) \in E(G)$, adding a path $P$ between $u$ and $v$ in such a way that $E(G) \cap E(P) = \emptyset$, $V(G) \cap V(P) =\{u,v\}$, and $|E(P)| \geq 2$ results in a polygonal 2-tree.
\end{definition}

\noindent
\noindent
A cycle consisting of $k$ edges is a \emph{$k$-gonal tree}. 
For a $k$-gonal 2-tree $G$ such that $(u,v) \in E(G)$, adding a path $P$ between $u$ and $v$ in such a way that $E(G) \cap E(P) = \emptyset$, $V(G) \cap V(P) =\{u,v\}$, and $|E(P)| = k-1$ results in a \emph{$k$-gonal 2-tree}.
For example, a \emph{2-tree} is a $3$-gonal tree.
The class of polygonal 2-trees is a subclass of planar graphs and it includes 2-connected outerplanar graphs and $k$-gonal trees.
2-trees, in other words $3$-gonal trees,  are  extensively studied in the literature. In particular, previous work on various flavors of counting and enumeration problems  on 2-trees  is compiled in \cite{Specification2Trees2002}.  Formulas for the number of labeled and unlabeled $k$-gonal trees with $r$ polygons (induced cycles) are computed in \cite{GilbertLabelledKgonal2004}.
The family of $k$-gonal trees with same number of vertices is claimed as a chromatic equivalence class by Chao and Li, and the claim has been proved by Wakelin and Woodal \cite{KohSPgraphs1996}.
The class of polygonal 2-trees is shown to be a chromatic equivalence class by Xu \cite{KohSPgraphs1996}.
Further, various subclasses of generalized polygonal 2-trees have been considered, and it has been shown that they also form a chromatic equivalence class  \cite{KohSPgraphs1996,OmoomiGPgraphs2003,PengGPgraphs97}.  The enumeration of outerplanar $k$-gonal trees is studied by Harary, Palmer and Read to solve a variant of the cell growth problem \cite{MartinOpKgonal2007}. 
Molecular expansion of the species of outerplanar $k$-gonal trees is shown in \cite{MartinOpKgonal2007}.
Also outerplanar $k$-gonal trees are of interest in combinatorial chemistry, as the structure of chemical compounds like catacondensed benzenoid hydrocarbons forms an outerplanar $k$-gonal tree.

%
\subsection{Our Results}


We state our main theorem.

\begin{theorem}
\label{Theorem_mainResult}
Given a polygonal 2-tree $G$ on $n$ vertices, a minimum average stretch spanning tree of $G$ can be obtained in $O(n \log n)$ time.
\end{theorem}

\noindent
A quick overview of our approach to solve \textsc{Mast} is presented in \textbf{Algorithm \ref{mastAlgo_1}} below. The detailed implementation is given in Section \ref{Section_ComputeMAST}.

\begin{algorithm2e}
\caption{An algorithm to find an MAST of a polygonal 2-tree $G$}
\label{mastAlgo_1}

$A \leftarrow \emptyset$\;
\lFor{ each edge $e \in E(G)$}{   $c[e] \leftarrow 0$\;}
\While{$G-A$ has a cycle}
{ 
 Choose an edge  $e$ from $G-A$, such that $e$ belongs to exactly one induced cycle in $G-A$ and $c[e]$ is minimum \;
Let $C$ be the induced cycle containing $e$ in $G-A$ \;
	\lFor{each $\hat{e} \in E(C) \setminus \{e\}$}
	{
		$c[\hat{e}] \leftarrow c[\hat{e}] + c[e] + 1$\;
	}
$A \leftarrow A \cup \{e\}$ \;

 } 
Return $G - A$;
\end{algorithm2e}

Due to the equivalence of \textsc{Mast} and \textsc{Mfcb} (shown in Lemma \ref{LemmaMASTMFCB}), our result implies the following corollary. 
For a set $\mathcal{B}$ of cycles in $G$, the size of $\mathcal{B}$, denoted by $\size(\mathcal{B})$, is the number of edges in $\mathcal{B}$ counted according to their multiplicity.

\begin{corollary}
 Given a polygonal 2-tree $G$ on $n$ vertices, a minimum fundamental cycle basis $\mathcal{B}$ of $G$ can be obtained in $O(n \log n + \size(\mathcal{B}))$ time.
\end{corollary}

\noindent
We characterize polygonal 2-trees using a kind
 of ear decomposition and present the structural properties of polygonal 2-trees that are useful in finding a minimum average stretch spanning tree (In Section \ref{SectionPropertiesPolygonal2Trees}). 
%
We then identify a set of edges in a polygonal 2-tree, called safe edges, whose removal results in a minimum average stretch spanning tree (In Section \ref{StructurePathsTreesMasts}).
We present an algorithm with necessary data-structures to  identify the safe set of edges efficiently and compute a minimum average stretch spanning tree in sub-quadratic time (In Section \ref{Section_ComputeMAST}).
We finally characterize polygonal 2-trees using cycle bases, which is of our independent interest (In Section \ref{SectionCyclebasis}).
\\

\noindent
A graph $G$ can be probabilistically embedded into its spanning trees with distortion $t$ if and only if  the multigraph obtained from $G$ by replicating its edges has a spanning tree with average stretch  at most $t$  (See \cite{Alon95graphTheoretic}).
It is easy to observe that, a spanning tree $T$ of $G$ is a minimum average stretch spanning tree for $G$ if and only if $T$ is a minimum average stretch spanning tree for a multigraph of $G$.
 As a consequence of our result, we have the following corollary.

\begin{corollary}
For a polygonal 2-tree $G$ on $n$ vertices, the minimum possible distortion of probabilistically embedding $G$ into its spanning trees can be obtained in $O(n\log n)$ time.
\end{corollary}

\noindent

\section{Graph Preliminaries}
We consider simple, connected, unweighted and undirected graphs. We use standard graph terminology from \cite{dbwestBook}. Let $G=(V(G),E(G))$ be a graph, where $V(G)$ and $E(G)$ denote the set of vertices and edges, respectively in $G$. We denote $|V(G)|$ by $n$ and $|E(G)|$ by $m$.
The union of graphs $G_1$ and $G_2$ is defined as a graph with vertex set $V(G_1) \cup V(G_2)$ and edge set $E(G_1) \cup E(G_2)$ and is denoted by $G_1 \cup G_2$.
The intersection of graphs $G_1$ and $G_2$ written as $G_1 \cap G_2$ is a graph with vertex set  $V(G_1) \cap V(G_2)$ and edge set $E(G_1) \cap E(G_2)$.
The removal of a set $X$ of edges from $G$ is denoted by $G-X$.
For a set $X \subset V(G)$, $G[X]$ denotes the induced graph on $X$.
An edge $e \in E(G)$ is a \emph{cut-edge} (bridge) if $G-e$ is disconnected.
A graph is \emph{2-connected} if it can not be disconnected by removing less than two vertices.
A \emph{2-connected component} of $G$ is a maximal 2-connected subgraph of $G$.


Let $T$ be a spanning tree of $G$.
An edge $e \in E(G) \setminus E(T)$ is a \emph{non-tree} edge of $T$.
For a  non-tree edge $(u,v)$ of $T$, a cycle formed by  the edge $(u,v)$ and the unique path between $u$ and $v$ in $T$ is referred to as a \emph{fundamental cycle}.
For an  edge $(u,v) \in E(G)$, \emph{stretch} of $(u,v)$ is the distance between $u$ and $v$ in $T$. The \emph{total stretch} of $T$ is defined as the sum of the stretches of all the edges in $G$.
We remark that there are slightly different definitions existing in the literature to refer the average stretch of a spanning tree.
 We use the definition in Equation \ref{EquationAveragStretch}, presented by Emek and Peleg in \cite{EmekSeriesParallelPE}, to refer the average stretch of a spanning tree.
Proposition 14 in \cite{LiebchenZooOfTreeSpanners} states that, $T$ is a minimum total stretch spanning tree of $G$ if and only if the set of fundamental cycles of $T$ is a minimum fundamental cycle basis of $G$. Then, we can have the following lemma.

\begin{lemma}
\label{LemmaMASTMFCB}
 Let $G$ be an unweighted graph and  $T$ be a spanning tree of $G$.
 $T$ is a minimum average stretch spanning tree of $G$ if and only if  the set of fundamental cycles of $T$ is a minimum fundamental cycle basis of $G$.
\end{lemma}

\noindent
We use the following convention crucially. A path is a connected graph in which there are two vertices of degree one and the rest of the vertices are of degree two. An edge can be considered as a connected graph consisting of single edge.

\begin{lemma}
\label{Lemma_subGraphSpanningTree}
Let $G'$ be a 2-connected component in an arbitrary graph $G$ and $T$ be a subgraph of $G$. \\
(a)~If $T$ is a spanning tree of $G$, then $T\cap G'$ is a spanning tree of $G'$.\\
(b)~If $T$ is a path in $G$, then $T \cap G'$ is a path.
\end{lemma}
\begin{proof}
We first prove the following claim: If $T$ is a tree, then $T \cap G'$ is a tree.\\

\noindent
Let $T'=T \cap G'$.
Suppose  $T'$ is not connected, then there exist two vertices $x$ and $y$ in $V(G')$ such that there is no path between $x$ and $y$ in $T'$.
Since $T$ is a tree, there is a path $P$ between $x$ and $y$ in $T$. 
Since $T'$ is not connected, we can observe that $V(P) \setminus V(G')$ contains at least one vertex, say, $u$. Further, the two edges incident on $u$ in $P$ are not in $G'$. Now we can obtain a graph $G' \cup P$ which is a 2-connected component in $G$. 
This contradicts the maximality of the 2-connected component $G'$. 
Therefore, $T'$ is connected. As $T'$ is acyclic, we conclude that $T'$ is a tree.

If $T$ is a spanning tree of $G$, then the set of vertices in $T \cap G'$ is $V(G')$. Therefore from the above claim, $T \cap G'$ is a spanning tree of $G'$. Thus (a) holds. Further, (b) also holds from the above claim.
\qed
\end{proof}

\noindent
\textbf{Special Graph Classes.} A \emph{partial 2-tree}  is a subgraph of a 2-tree.  
A graph is a \emph{series-parallel} graph, if it can be obtained from an edge, by repeatedly duplicating an edge between its end vertices or replacing an edge by a path.
An alternative equivalent definition for series-parallel graphs is given in \cite{Eppstein92SP}.

\section{Structure of Polygonal 2-trees and Computation of Induced Cycles}
\label{SectionPropertiesPolygonal2Trees}
In this section, we present crucial structural properties of polygonal 2-trees in Lemma \ref{Lemma_Polygonal2TreeProperty}.
This lemma will be used significantly in proving the correctness of our algorithm.
Another major result in this section is Theorem \ref{TheoremPolygonal2TreeNiceEarDecomp}, which computes a kind of ear decomposition for polygonal 2-trees. This helps in obtaining an efficient algorithm to solve \textsc{Mast}.
The notion of open ear decomposition is well known to characterize 2-connected graphs \cite{dbwestBook}.
An \emph{open ear decomposition} of $G$ is a partition of $E(G)$ into a sequence $(P_0, \ldots, P_k)$ of edge disjoint graphs  called  \emph{ears} such that,
 \begin{enumerate}
\item For each $i \geq 0$, $P_i$ is a path.
\item For each $i \geq 1$,  end vertices of $P_i$ are distinct and the internal vertices of $P_i$ are not in $P_0 \cup \ldots \cup P_{i-1}$.

 \end{enumerate}

\noindent
Further, a restricted version of open ear decomposition called \emph{nested ear decomposition} is used to characterize series-parallel graphs \cite{Eppstein92SP}. An open ear decomposition $(P_0, \ldots, P_k)$ of $G$ is said to be \emph{nested} if it satisfies the following properties: 
\begin{enumerate}
\item For each $i \geq 1$, there exists  $j <i$, such that the end vertices of path $P_i$ are in $P_j$.
\item Let the end vertices of $P_i$ and $P_{i'}$ are in $P_{j}$,  where $0 \leq j < i,i'\leq k$ and $i \neq i'$.
Let $Q_{i} \subseteq P_{j}$ be the path between the end vertices of $P_i$ and $Q_{i'} \subseteq P_j$ be the path between the end vertices of $P_{i'}$.
Then  $E(Q_i) \subseteq E(Q_{i'})$ or $E(Q_{i'}) \subseteq E(Q_{i})$ or $E(Q_i) \cap E(Q_{i'}) = \emptyset$.
\end{enumerate}
We define nice ear decomposition to characterize polygonal 2-trees and we show how it helps in efficiently computing the induced cycles. A nested ear decomposition $(P_0, \ldots, P_k)$ is said to be  \emph{nice} if it has the following property: $P_0$ is an edge and 
for each $i \geq 1$, if $x_i$ and $y_i$ are the end vertices of $P_i$, then there is some $j <i$, such that $(x_i,y_i)$ is an edge in $P_j$. A nice ear decomposition of a polygonal 2-tree is shown in Fig \ref{figNice}.
Definition \ref{DefinitionPolygonal2Tree} naturally gives a nice ear decomposition for polygonal 2-trees. Further, a unique polygonal 2-tree can be constructed easily from a nice ear decomposition. Thus we have the following observation.

\begin{observation}
\label{ObservationNiceEarDecomp}
 A graph $G$ is a polygonal 2-tree if and only if $G$ has a nice ear decomposition.
\end{observation}


\begin{figure}[!ht]
\centering
  \subfloat[\label{figNice}]{%
    \includegraphics[width=0.28\textwidth]{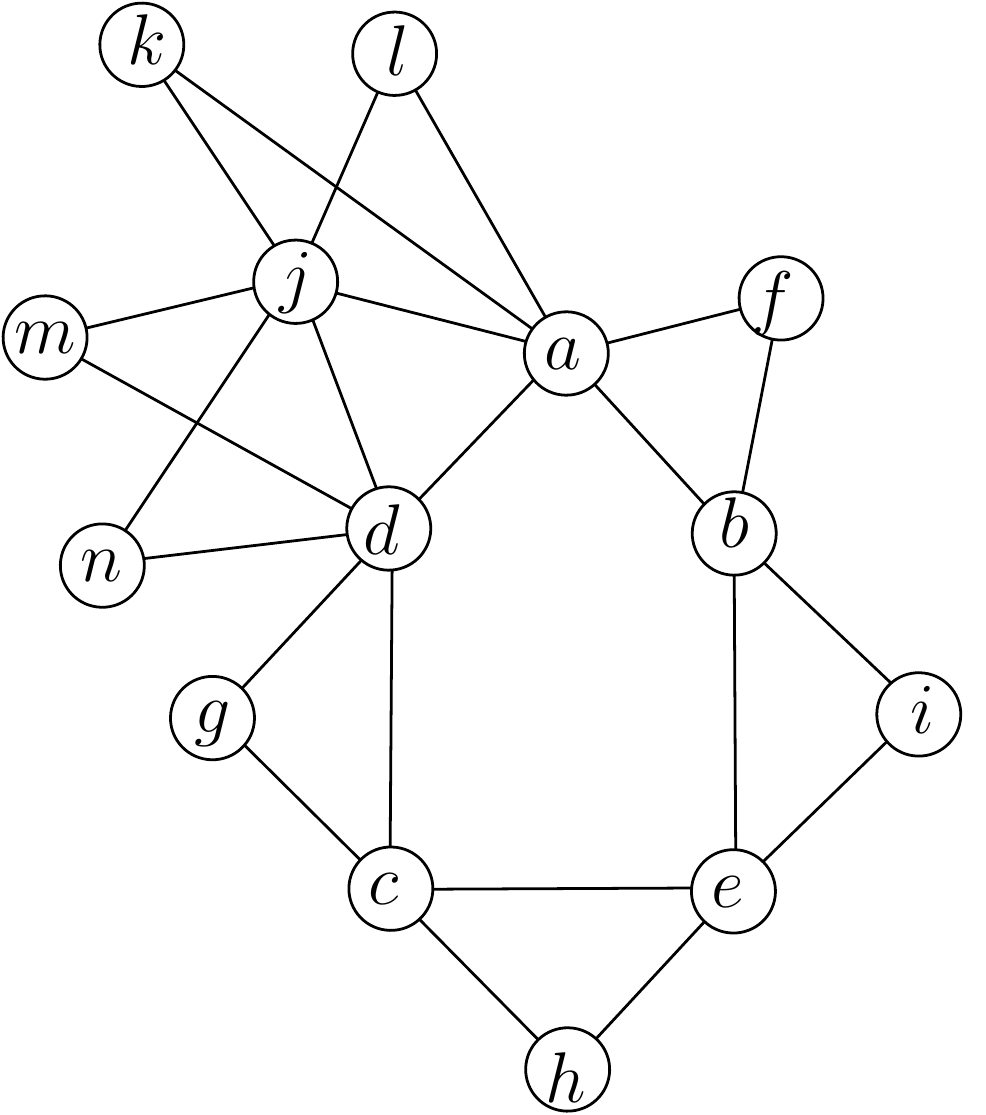}
  }
  \subfloat[\label{figBound}]{%
    \includegraphics[width=0.34\textwidth]{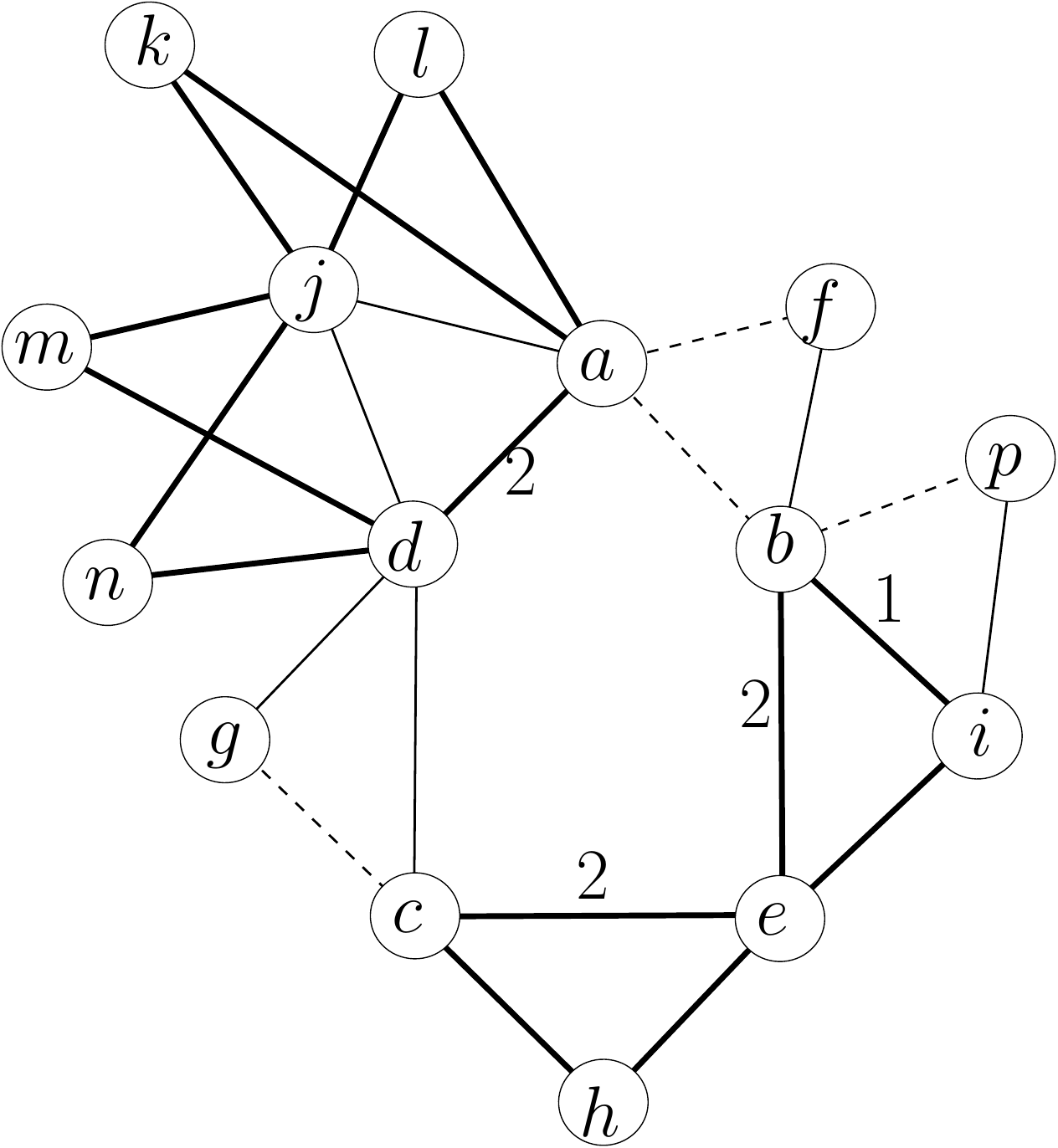}
  }
\subfloat[\label{figPathStructure}]{%
    \includegraphics[width=0.34\textwidth]{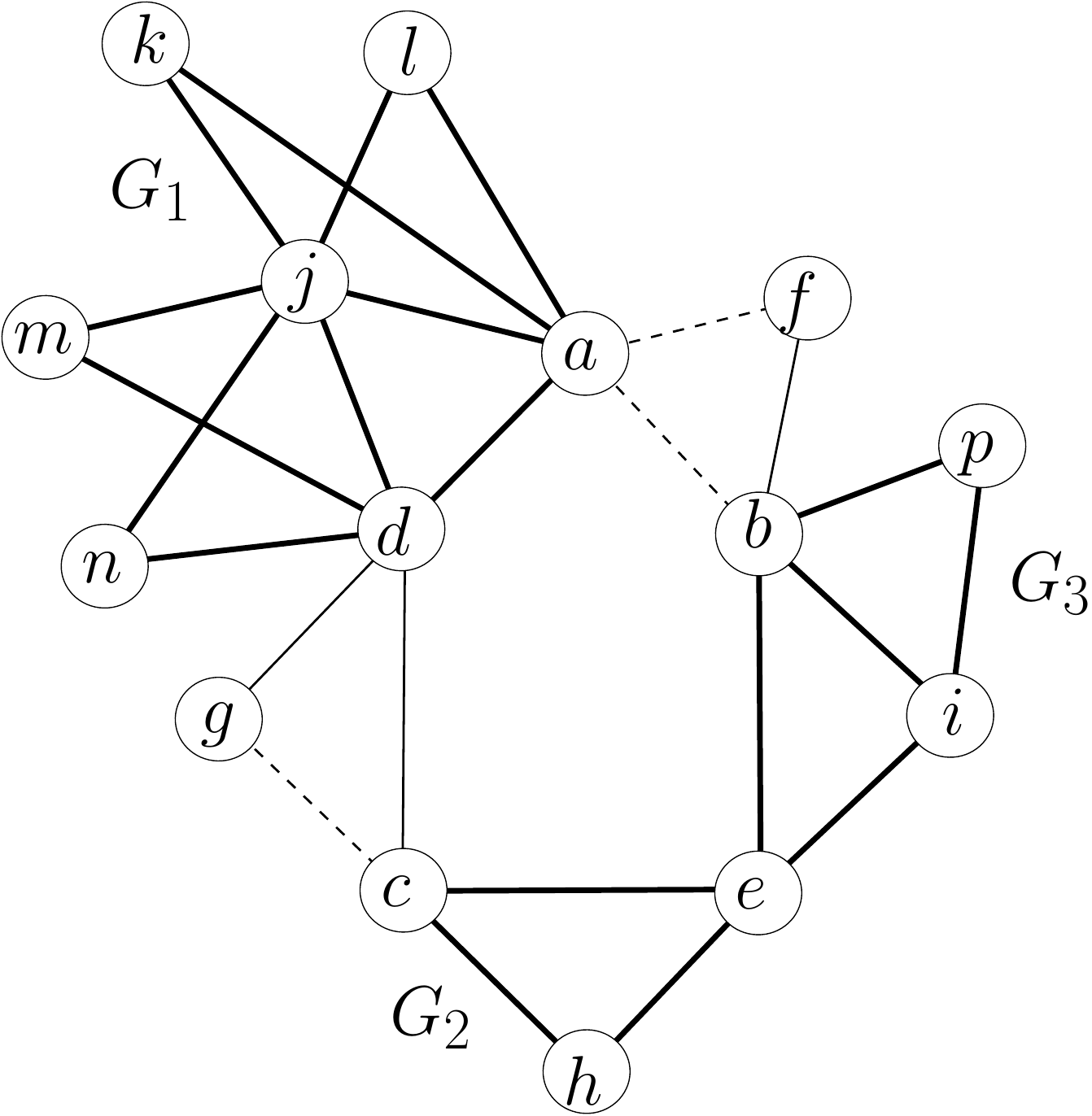}
  }
  \caption{\textbf{(a)} For the polygonal 2-tree shown, $(P_0, \ldots, P_{10})$ is a nice ear decomposition, where $P_0 = (a,b)$, $P_1=(a,d,c,e,b)$, $P_2=(a,f,b)$, $P_3=(c,g,d)$, $P_4=(c,h,e)$, $P_5=(b,i,e)$, $P_6=(a,j,d)$, $P_7=(a,k,j)$, $P_8 = (a,l,j)$, $P_9 = (d,m,j)$, and $P_{10}=(d,n,j)$ are paths in $G$.
\textbf{(b)} For the polygonal 2-tree $G$ shown, let $A = \{ (a,f), (a,b), (b,p),(c,g)\}$. The edges in $\bound(A,G)$ are shown in thick. $\Sup((a,d)) =  \Sup((b,e)) =\Sup((c,e)) = \{ (a,b), (a,f) \}$ and $\Sup((b,i)) = \{(b,p)\}$. $\cost((a,d)) = 2$, $\cost((c,e))=2$, $\cost((b,e))=2$, $\cost((b,i)) = 1$ and for the rest of the edges in $\bound(A,G)$, $\cost$ is zero.
\textbf{(c)} For the polygonal 2-tree $G$ shown, let $A=\{(a,b),(a,f),(c,g)\}$. The 2-connected components $G_1$, $G_2$ and $G_3$ in $G-A$ are polygonal 2-trees.
Let $P = (a,d,c,e,b,f)$ be the shortest path between vertices $a$ and $f$ in $G-A$. Then $P$ intersects exactly with one edge in the graphs $G_1$, $G_2$ and $G_3$.
}
\end{figure}

\noindent
In the following lemmas, we present results from the literature that establish polygonal 2-trees as a subclass of 2-connected partial 2-trees, which we formalize in Lemma \ref{Lemma_Polygonal2TreeProperty}.  
\begin{lemma}[\textbf{Theorem 42 in \cite{partialKarboreturm}}]
\label{Lemma2ConnectedPartial2TreeEqualSPgraphs}
A graph $G$ is a partial 2-tree if and only if every 2-connected component of $G$ is a series-parallel graph.
\end{lemma}

\noindent
According to Lemma \ref{Lemma2ConnectedPartial2TreeEqualSPgraphs}, 2-connected series-parallel graphs and 2-connected partial 2-trees are essentially same.
\begin{lemma} [\textbf{Lemma 1, Lemma 7 and Theorem 1 in \cite{Eppstein92SP}}]
\label{Lemma2ConnectedPartial2TreeNested}
 A graph $G$ is 2-connected if and only if $G$ has a open ear decomposition in which the first ear is an edge.
Further, for a 2-connected series-parallel graph, every open ear decomposition is nested.
A graph is series-parallel if and only if it has a nested ear decomposition.
\end{lemma}

\noindent
The above lemma implies that every 2-connected partial 2-tree has a nested ear decomposition \emph{starting} with an edge (first ear is an edge) and vice versa.
We strengthen the first part of this result in Lemma \ref{LemmaSpecialNestedEarDecomposition}. 

\subsection{Necessary and Sufficient Conditions}

From Propositions 1.7.2 and 12.4.2 in \cite{Diestel2010}, partial 2-trees do not contain a $K_4$-subdivision (as a subgraph). The following lemma presents a few necessary properties of polygonal 2-trees, which are useful in the rest of the paper.

\begin{lemma}
\label{Lemma_Polygonal2TreeProperty}
Let $G$ be a polygonal 2-tree. Then,\\
(a)~$G$ is a  2-connected partial 2-tree and $G$ does not contain a $K_{4}$-subdivision.\\
(b)~Any two induced cycles in $G$ share at most one edge and at most two vertices.\\
(c)~For $u,v \in V(G)$ such that $(u,v) \notin E(G)$, $G-\{u,v\}$ has at most two components.
\end{lemma}

\noindent
\begin{proof}
From Lemma \ref{Lemma2ConnectedPartial2TreeNested}, a graph is a 2-connected partial 2-tree if and only if it has a nested ear decomposition starting with an edge.
From Observation \ref{ObservationNiceEarDecomp}, a graph is a polygonal 2-tree if and only if it has a nice ear decomposition.
Observe that nice ear decomposition is a restricted version of nested ear decomposition. Therefore, a polygonal 2-tree is a 2-connected partial 2-tree. Recall that partial 2-trees do not contain $K_4$-subdivision as a subgraph. It follows that polygonal 2-trees do not contain $K_4$-subdivision as a subgraph.

We now prove that any two induced cycles in $G$ share at most one edge and at most two vertices.
Let $D =(P_0,\ldots,P_k)$ be a nice ear decomposition of $G$. The proof is by induction on the number of ears in $G$. If the number of ears in $D$ is one, then the claim is trivially true.
If the number of ears in $D$ is at least two, then
we remove the internal vertices of $P_k$ from $G$ and let $G'$ be the resultant graph.
 Let $D'=(P_0,\ldots, P_{k-1})$.
As $G'$ is a polygonal 2-tree and $D'$ is a nice ear decomposition of $G'$, inductively $G'$ satisfies (b).
Let $u$ and $v$ be the end vertices of $P_k$.
For the induced cycle $C = P_k \cup (u,v)$, $C \cap G'$ is $(u,v)$.  Therefore $C$ has at most one edge and two vertices in common with the induced cycles in $G'$. Because $V(P_k) \cap V(G') =\{u,v\}$ and $(u,v) \in E(G')$, $C$ is the only induced cycle not in $G'$. Hence, any two induced cycles in $G$ share at most one edge and at most two vertices. 


We now prove the last claim of this lemma. The proof is by contradiction. We assume that the removal of vertices $u$ and $v$ from $G$ such that $(u,v) \notin E(G)$ disconnects $G$ into at least three components $G_1$, $G_2$ and $G_3$. Note that $\{u,v\}$ is a minimal vertex separator in $G$, because $G$ is 2-connected.
It follows that, for each $1 \leq i \leq 3$, there is an induced path $P_i$ between $u$ and $v$ in $G$, such that the internal vertices of $P_i$ are in $G_i$ and $|E(P_i)| \geq 2$.  We have induced cycles $C_1 = P_1 \cup P_3$ and $C_2 = P_2 \cup P_3$ that share at least two edges, which contradicts that any two induced cycles in $G$ have at most one edge common.
\qed
\end{proof}

\noindent We now present a sufficient condition for a graph to be a polygonal 2-tree.
\begin{lemma}
\label{Lemma_sufficientPolygonal2Tree}
 If $G$ is a 2-connected partial 2-tree and every two induced cycles in $G$ share at most one edge, then $G$ is a polygonal 2-tree.
\end{lemma}
\begin{proof}
On the contrary, assume that  $G$ is not a polygonal 2-tree.
By Lemma \ref{LemmaSpecialNestedEarDecomposition}, since $G$ is a 2-connected partial 2-tree,  $G$ has a nested ear decomposition $D=(P_0, \ldots, P_k)$ such that $P_0$ is an edge and for each $i \geq 1$, $|E(P_i)| \geq 2$. 
Since $G$ is not a polygonal 2-tree, $D$ is not a nice ear decomposition.
Therefore, there exists an index $i \in \{1,\ldots,k\}$ with the property that, let $u$ and $v$ be the end vertices of $P_i$, then for every $j <i$, $(u,v) \notin E(P_j)$.
For every $j \geq i$, since $|E(P_j)| \geq 2$ and no internal vertex of $P_j$ is in $P_1, \ldots, P_{j-1}$, $(u,v) \notin E(P_j)$.
Thereby $(u,v) \notin E(G)$.
As $P_0 \cup \ldots \cup P_{i-1}$ is 2-connected, there exist two internally vertex disjoint paths  $P'_1$ and $P'_2$ between $u$ and $v$.
Since $(u,v) \notin E(G)$, $P'_1$, $P'_2$ and $P_i$ are internally vertex disjoint paths and each of these paths have at least one internal vertex.
Due to Lemma \ref{LemmaAtLeastThreeComponents}, for $1 \leq i \neq j \leq 3$, there is
no path between any internal vertex in $P_i$ and any internal vertex in $P_j$ that excludes the vertices $u$ and $v$.
Now we have two induced cycles $P'_1 \cup P_i$ and $P'_2 \cup P_i$ that share at least two edges. This contradicts the premise of the lemma. Therefore, $G$ is a polygonal 2-tree.
\qed
\end{proof}

\subsection{Computation of Induced Cycles in Polygonal 2-trees}

\noindent
Our algorithm will perform several computations on the induced cycles of a polygonal 2-tree. It is therefore
 important to obtain the set of induced cycles in a polygonal 2-tree in linear time. We prove this in Theorem \ref{TheoremPolygonal2TreeNiceEarDecomp}.
The proof is based on the following two lemmas and 
 a linear-time algorithm for obtaining an open ear decomposition \cite{Ramachandran92parallelopen}.




\begin{lemma}
\label{LemmaAtLeastThreeComponents}
Let $G$ be a partial 2-tree and
let $P_1,P_2$ and $P_3$ be three internally vertex disjoint paths between vertices $u$ and $v$ in $G$ such that $(u,v) \notin E(G)$. 
Then $G - \{u,v\}$ has at least three components.
\end{lemma}
\begin{proof}
Assume that $G-\{u,v\}$ has at most two components. Then without loss of generality, there is a path $P$ between  $x \in V(P_1)$ and $y \in V(P_2)$, such that
$V(P) \cap V(P_3) = \emptyset$, $V(P) \cap V(P_1) = \{x\}$, and $V(P) \cap V(P_2) = \{y\}$.
Then there is a $K_4$-subdivision on the vertices $x,y,u$ and $v$ in $G$.
It contradicts that a partial 2-tree does not contain a $K_4$-subdivision. Thus $G - \{u,v\}$ has at least three components.
\qed
\end{proof}

\begin{lemma}
\label{LemmaSpecialNestedEarDecomposition}
 Let $G$ be a 2-connected partial 2-tree. Then there exists a nested ear decomposition $(P_0, \ldots, P_k)$  of $G$, such that $P_0$ is an edge and for each $i \geq 1$, $|E(P_i)| \geq 2$.
\end{lemma}
\begin{proof}
From Lemma \ref{Lemma2ConnectedPartial2TreeNested}, $G$ has a nested ear decomposition $D=(P_0, \ldots, P_k)$ such that $P_0$ is an edge. 
Suppose $D$ does not satisfy the given constraint, then we update $D$  as follows, so that the resultant nested ear decomposition satisfies the given constraint.
Let $P_i$ be the first path in the sequence $D$, such that $|E(P_i)| = 1$, where $i \geq 1$.
Let $P_j$ be the first path in the sequence $D$, such that the end vertices of $P_i$ are in $P_j$, where $j <i$.
Let $x$ and $y$ be the end vertices of $P_i$.
We obtain new paths $P_{i'}$ and $P_{j'}$ from $P_i$ and $P_j$ as follows:
 $P_{i'}$ is the path between $x$ and $y$ in $P_{j}$ 
 and $P_{j'}$ is $P_{j} \cup P_{i} - X$, where $X$ is the set of internal vertices in $P_{i'}$.
We replace $P_j$ with $P_{j'}$, delete $P_i$ and add $P_{i'}$ immediately after $P_{j'}$. By performing the update steps mentioned above for at most $k-2$ times, we obtain a nested ear decomposition that satisfies the desired constraint.
\qed
\end{proof}

\noindent In the lemma below, we show 
that a nested ear decomposition as in Lemma \ref{LemmaSpecialNestedEarDecomposition} is a nice ear decomposition for polygonal 2-trees and it can be computed in linear time.

\begin{theorem}
\label{TheoremPolygonal2TreeNiceEarDecomp}
Let $G$ be a polygonal 2-tree on $n$ vertices. Let $D$ be a nested ear decomposition of $G$ as in Lemma \ref{LemmaSpecialNestedEarDecomposition} and $\mathcal{B}$ be the set of induced cycles in $G$.
Then $D$ is a nice ear decomposition. Further, $D$ and $\mathcal{B}$ can be computed in linear time and $\size({\mathcal{B})}$ is $O(n)$.
 \end{theorem}

\begin{proof}
On the contrary, assume that  $D$ is not a nice ear decomposition.
By Lemma \ref{LemmaSpecialNestedEarDecomposition}, $G$ has a nested ear decomposition $D=(P_0, \ldots, P_k)$ such that $P_0$ is an edge and for each $i \geq 1$, $|E(P_i)| \geq 2$. 
Since $D$ is not a nice ear decomposition,
there exists an index $i \in \{1,\ldots,k\}$ with the property that, let $u$ and $v$ be the end vertices of $P_i$, and for every $j <i$, $(u,v) \notin E(P_j)$.
For every $j \geq i$, since $|E(P_j)| \geq 2$ and no internal vertex of $P_j$ is in $P_1, \ldots, P_{j-1}$, $(u,v) \notin E(P_j)$.
Thereby $(u,v) \notin E(G)$.
As $P_0 \cup \ldots \cup P_{i-1}$ is 2-connected, there exist two internally vertex disjoint paths  $P'_1$ and $P'_2$ between $u$ and $v$.
Since $(u,v) \notin E(G)$, $P'_1$, $P'_2$ and $P_i$ are internally vertex disjoint paths.
Due to Lemma \ref{LemmaAtLeastThreeComponents}, $G-\{u,v\}$ has at least three components, which is a contradiction to Lemma \ref{Lemma_Polygonal2TreeProperty}.(c).

We now prove that a nice ear decomposition of $G$ can be obtained in $O(n)$ time.
First obtain an open ear decomposition $D'$ starting with an edge by using linear-time algorithm in \cite{Ramachandran92parallelopen}. We then apply Lemma \ref{LemmaSpecialNestedEarDecomposition} on $D'$. This takes linear time, because the number of ears in $D'$ is at most $n$ and we spend only a constant amount of time at each ear. From the first part of this lemma, the resultant ear decomposition is a nice ear decomposition. Also note that $|E(G)| \leq 2n-3$. Thus a nice ear decomposition $(P_0, \ldots, P_k)$ is computed in $O(n)$ time.

From the nice ear decomposition $D=(P_0, \ldots, P_k)$ of $G$, we now present a linear-time procedure to obtain the set of induced cycles in $G$. Since $P_0$ is an edge, $C_1 = P_0 \cup P_1$ is an induced cycle in $G$. For every $i \geq 2$, let $x_i$ and $y_i$ be the end vertices of $P_i$, we obtain an induced cycle $C_i = P_i \cup (x_i,y_i)$ in $G$. 
Observe that $C_1, \ldots, C_k$ are the only induced cycles in $G$. This can be proved easily by applying induction on the number of ears in $D$.
The number of ears in $D$ is at most $n$. Thus the set of induced cycles in $G$ can be obtained in $O(n)$.  The ears $P_0, \ldots, P_k$ is a partition of $E(G)$. Therefore, $|E(C_0)| + \ldots + |E(C_k)| \leq |E(G)| + n$. Thus $\size(\mathcal{B})$ is $O(n)$.
\qed
\end{proof}

\section{Structure  of Paths, Trees and MASTs in Polygonal 2-trees}
\label{StructurePathsTreesMasts}
For the rest of the paper, $G$ denotes a polygonal 2-tree. In this section we design an 
iterative procedure to delete a subset of edges from a polygonal 2-tree, so that the graph on the remaining edges is a minimum average stretch spanning tree. 
This result is shown in Theorem \ref{Theorem_safeEdges}.

\noindent
{\em Important Definitions:} We introduce some necessary definitions on polygonal 2-trees. 
Two induced cycles in $G$ are \emph{adjacent} if they share an edge. 
An edge in $G$ is \emph{internal} if it is part of at least two induced cycles; otherwise it is \emph{external}. 
An induced cycle in $G$ is \emph{external} if it has an external edge; otherwise it is \emph{internal}.
A fundamental cycle of a spanning tree, created by a non-tree edge is said to be \emph{external} if the associated non-tree edge is external.
For a cycle $C$ in $G$, the \emph{enclosure} of $C$ is defined as $G[V(C)]$ and is denoted by $Enc(C)$. 
A set $A\subseteq E(G)$ consisting of $k$ ($\geq 0$) edges is said to be an \emph{iterative} set for $G$ if the edges in $A$ can be ordered as $e_1, \ldots, e_k$ such that $e_1$ is external and not a bridge in $G$, and for each $ 2 \leq i \leq k$, $e_i$ is external and not a bridge in $G- \{ e_1, \ldots, e_{i-1}\}$.
Let $A$ be an iterative set of edges in $G$.
For every  edge $(u,v) \in A$, both $u$ and $v$ are not present in the same 2-connected component in $G-A$.
We define $\bound(A,G)$ to be the set of external edges in $G - A$ that are not bridges.
For an edge $e \in \bound(A,G)$, $G_e$ denotes the 2-connected component in $G-A$ that has $e$.
The following definition is illustrated  in Fig \ref{figBound}.

\begin{definition}
\label{def_SupportCost}
Let $A$ be an iterative set of edges in $G$ and
$e \in  bound(A,G)$. The support of $e$ is defined as
$\{ (u,v) \in A \mid$  there is a path $P$ joining $u$ and $v$ in $G-A$ such that $P \cap G_e = e \}$
and is denoted by $\Sup(e)$.
The $\cost(e)$ is defined as $|\Sup(e)|$.
\end{definition}


\subsection{Structural Properties of Paths}

\noindent In the following lemmas we 
present a result on the structure of paths connecting the end vertices of edges in an iterative set $A$.  This is useful to set up an iterative approach for computing a minimum average stretch spanning tree.
We apply the necessary properties of polygonal 2-trees (cf. Lemma \ref{Lemma_Polygonal2TreeProperty}) and sufficient condition for a graph to be a polygonal 2-tree (cf. Lemma \ref{Lemma_sufficientPolygonal2Tree}) in the proofs of the following lemmas.

\begin{lemma}
\label{Lemma_shareSingleEdgeNew}
Let $A$ be an iterative set of edges for $G$ and $(u,v) \in A$, $P$ be a path joining $u$ and $v$ in $G-A$, $G'$ be a 2-connected component in $G-A$ that has at least two vertices from $P$, and 
let $P' = P \cap G'$ be a path with end vertices $x$ and $y$. Then the following are true:\\
$(a)$~$(x,y) \in E(G')$.\\
$(b)$~If $P$ is a shortest path, then $P'$ is an edge.\\
$(c)$~Every 2-connected component in $G-A$ is a polygonal 2-tree.
\end{lemma}
\begin{proof}
To show that $(x,y) \in E(G')$, assume to the contrary that $(x,y) \notin E(G')$.
Since $G'$ is 2-connected, there exist two internally vertex disjoint paths $P_1$ and $P_2$ between $x$ and $y$ in $G'$.
Since  $A$ is an iterative set of edges for $G$ and $(u,v) \in A$, $|\{u,v\} \cap V(G')| \leq 1$.
It follows that $P' \subset P$. Then from the cycle $P \cup (u,v)$, we choose a path $P_3$ joining $x$ and $y$, in such a way that $P_3$ is edge disjoint from $P'$.
Consequently, none of the internal vertices in $P_3$ are from $G'$.
Therefore, $P_1,P_2$ and $P_3$ are internally vertex disjoint paths joining $x$ and $y$ that have at least one internal vertex. 
By Lemma \ref{LemmaAtLeastThreeComponents}, $G-\{x,y\}$ has at least three components.
Then the contrapositive of Lemma \ref{Lemma_Polygonal2TreeProperty}.(c) implies that $G$ is not a polygonal 2-tree. This contradicts that $G$ is a polygonal 2-tree. Thus $(x,y) \in E(G')$.

If $P$ is a shortest path and $P'$ is not an edge, then we can replace $P'$ in $P$ by $(x,y)$ and obtain a path shorter than $P$. 
Therefore, $P'$ is an edge.

We now prove the third claim of this lemma. Let $H$ be a 2-connected component in $G-A$. From Lemma \ref{Lemma_Polygonal2TreeProperty}.(a), $G$ is a partial 2-tree. Thereby $H$ is a 2-connected partial 2-tree. 
Since $A$ is an iterative set, the edges in $A$ can be ordered as $e_1, \ldots, e_k$, such that $e_1$ is external and not a bridge in $G$ and for each $2 \leq i \leq k$, $e_i$ is external and not a bridge in $G-\{e_1, \ldots, e_{i-1} \}$. We delete the edges in $A$ from $G$ one by one, in the order $e_1, \ldots, e_k$. Observe that each time, when an edge $e_i$ is deleted, exactly one induced cycle is destroyed and no new induced cycles are created. 
 Also we know that any two induced cycles in $G$ share at most one edge. Consequently, any two induced cycles in $H$ share at most one edge. Therefore, Lemma \ref{Lemma_sufficientPolygonal2Tree} implies that $H$ is a polygonal 2-tree.
\qed 
\end{proof}

\noindent
Lemma \ref{Lemma_shareSingleEdgeNew} is illustrated in Fig \ref{figPathStructure}.


\noindent

\begin{lemma}
\label{LemmaSupportBothDirections}
Let $A$ be an iterative set of edges for $G$. Then $(u,v) \in \Sup(e)$ if and only if there is a shortest path $P$ joining $u$ and $v$ in $G-A$ and $P$ has $e$. 
\end{lemma}
We use the following lemma to prove Lemma \ref{LemmaSupportBothDirections}.

\begin{lemma}
\label{Lemma_uniquePath}
 Let $P$ be a path with end vertices $u$ and $v$ in $G$.
Let $G_1, \ldots, G_r$ be the 2-connected components in $G$ from which $P$ has at least two vertices.
For each $1 \leq i\leq r$, let $P_i$ be a shortest path joining the end vertices of $G_i \cap P$.
Let $P'$ be the path obtained from $P$ by replacing every $G_i \cap P$ with $P_i$.
 Then $P'$ is a shortest path joining $u$ and $v$ in $G$.
\end{lemma}
\begin{proof}
Assume that there exists a path $P''$ joining $u$ and $v$ in $G$ such that $|E(P'')| < |E(P')|$.
For each $i$, let $x_i$ and $y_i$ be the end vertices of $P_i$.
The set of edges in $P'$ that are bridges in $G$ are definitely in $P''$.
Therefore, there exist an $1 \leq i \leq r$, such that
the subpath between $x_i$ and $y_i$ in $P''$ is shorter than $P_i$.
This contradicts that $P_i$ is a shortest path joining $x_i$ and $y_i$.
\qed
\end{proof}

\noindent The above lemma holds for arbitrary graphs.

\begin{proof}[of Lemma \ref{LemmaSupportBothDirections}]
($\Rightarrow$)
Let $(u,v) \in \Sup(e)$. By the definition of $\Sup(e)$, there is a path $P'$
joining $u$ and $v$ in $G-A$ such that $G_e \cap P'$ is $e$. 
Let $G_1, \ldots, G_r$ be the 2-connected components in $G-A$ from which $P'$ has at least two vertices.
For each $1 \leq i \leq r$,  by Lemma \ref{Lemma_subGraphSpanningTree}.(b), $P_i = G_i \cap P'$  is a path; let $x_i$ and $y_i$ be the end vertices of $P_i$; due to Lemma \ref{Lemma_shareSingleEdgeNew}.(a), $(x_i,y_i)\in E(G_i)$.
Let $P$ be the path obtained from $P'$ after replacing every $P_i$ by $(x_i,y_i)$.
Since $G_e \cap P'$ is $e$, $P$ has $e$.
From Lemma \ref{Lemma_uniquePath}, $P$ is a shortest path joining $u$ and $v$ in $G-A$ and $P$ has $e$.

($\Leftarrow$) 
Let $P$ be a shortest path joining $u$ and $v$ in $G-A$ such that $P$ has $e$.
Let $G_e$ be a 2-connected component containing $e$ in $G-A$.
Since $P$ has $e$, $G_e$ has at least two vertices from $P$.
 From Lemma \ref{Lemma_shareSingleEdgeNew}.(b), $G_e \cap P$ is an edge. 
Further, $G_e \cap P$ is $e$.
Thus $(u,v) \in \Sup(e)$.
\qed
\end{proof}

\subsection{Structural Properties of Spanning Trees}

\begin{lemma}
\label{Lemma_helpingLemma}
Let $T$ be a spanning tree of $G$ and $e$ be an external edge in $G$ such that $e \in E(T)$.
For the spanning tree $T$, let $C_{min}$ be the smallest fundamental cycle containing $e$ and
let $C_{max}$ be a largest fundamental cycle containing $e$.
Let $e'$ and $e''$ be the non-tree edges associated with $C_{min}$ and $C_{max}$, respectively.
Then, (a)~$e''$ is an external edge  (b)~$Enc(C_{min}) \subseteq Enc(C_{max})$. 
\end{lemma}
We use the following lemma to prove Lemma \ref{Lemma_helpingLemma}.

\begin{lemma}
\label{Lemma_InducedCycleOutside}
Let $T$ be an arbitrary spanning tree of $G$. Let $C$ be a fundamental cycle of $T$ formed by a non-tree edge $(x,y)$ in $G$. 
Let $C_1$ be an induced cycle containing $(x,y)$ in $Enc(C)$ and $C_2$ be another induced cycle containing $(x,y)$ in $G$. Then \\
(a) $V(C) \cap V(C_2) = \{x,y\}$.
(b)For vertices $u \in V(C) \setminus \{x,y\}$ and $v \in V(C_2) \setminus \{x,y\}$, any path joining $u$ and $v$ in $G$ goes through $x$ or $y$.
\end{lemma}
\begin{proof}
Assume that $V(C) \cap V(C_2)$ has a vertex that is different from $x$ and $y$.
In the path consisting of at least two edges from $x$ to $y$ in $C_2$, let $z$ and $z'$ be the first and last vertices from $C$, respectively.
From Lemma \ref{Lemma_Polygonal2TreeProperty}.(b), we know that any two induced cycles in a polygonal 2-tree share at most two vertices. Thus $V(C_1) \cap V(C_2) = \{x,y\}$ and $z,z' \notin V(C_1)$.
 Let $(x',y')$ be an edge in $C_1$ such that $(x',y') \neq (x,y)$. Further, without loss of generality, assume that $x' \neq x$. 
The graph $C \cup C_2$ is shown in Fig \ref{figK4One}, where
the edges in $C$ and $C_2$ other than $(x,y)$ are shown by solid edges and bold edges, respectively.
There is a $K_4$-subdivision in $C \cup C_2$ on the vertices $\{x,y,z,x'\}$,
because of the following six paths that are internally vertex disjoint: the edge $(x,y)$; the path joining $x'$ and $x$ in $C_1$ without going through $y$; the path joining $x'$ and $y$ in $C_1$ without going through $x$; the path between $z$ and $x$ in $C_2$ without going through  $y$; the path between $z$ and $y$ (via $z'$) in $C_2$ without going through $x$; the path between $z$ and $x'$ in $C$ without going through $y'$.  
This contradicts that $G$ does not contain a $K_{4}$-subdivision. Thus $V(C) \cap V(C_2) =\{x,y\}$.


We now prove the second part of the lemma. Let $P$ be a path that joins vertices $u$ and $v$, such that $x,y \notin V(P)$.
 In the sequence of vertices in $P$ from $u$ to $v$, let $u'$ be the last vertex in $C$ and $v'$ be the  first subsequent vertex in $C_2$. 
 Let $P' \subseteq P$ be the path joining the vertices $u'$ and $v'$.
From the first part of this lemma, $x$ and $y$ are the only vertices common in $C$ and $C_2$. Thereby $u'$ is different from $v'$.
It follows that the edges in $P'$ are disjoint from the edges in $C \cup C_2$.
Now, we consider the graph $H = C \cup C_2 \cup P'$. The subgraph $H$ of $G$, shown in Fig \ref{figK4Two}, is a $K_4$-subdivision on the vertices $x,y,u',v'$, because for every two vertices in $\{x,y,u',v'\}$, there is an internally vertex disjoint path. We have a contradiction, as $G$ does not contain a $K_{4}$-subdivision. Hence the lemma.
\qed
\end{proof}



\begin{figure}[!ht]
\centering
  \subfloat[ \label{figK4One}]{%
    \includegraphics[width=0.2\textwidth]{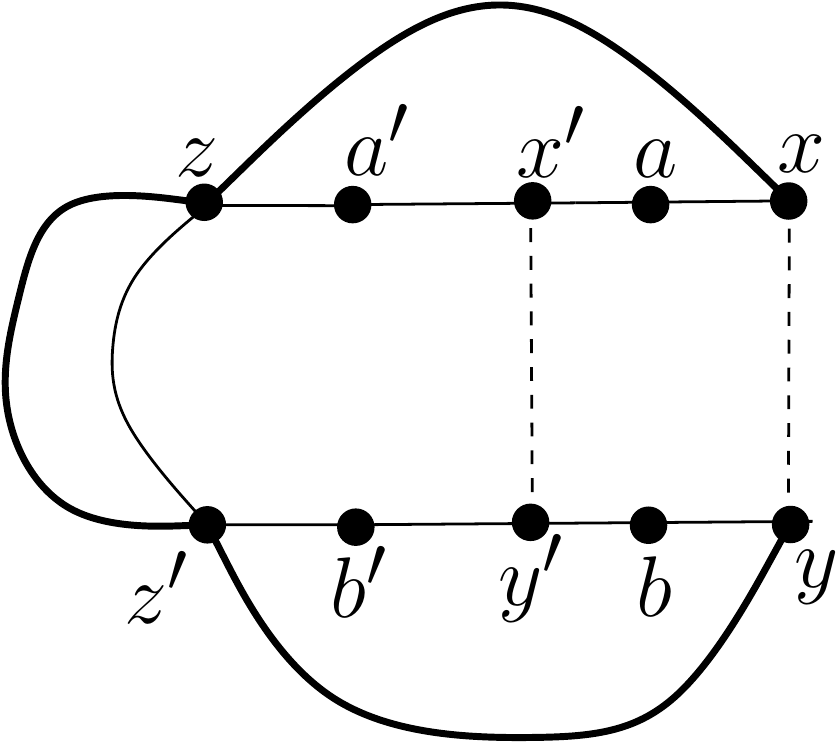}
  }
  \quad \quad \quad
  \subfloat[ \label{figK4Two}]{%
    \includegraphics[width=0.22\textwidth]{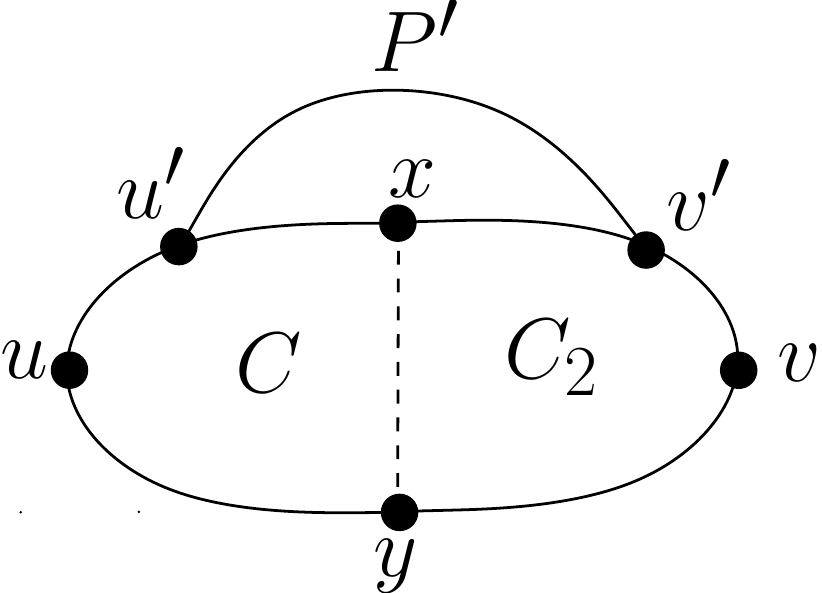}
  }
  \caption{\textbf{(a)} In the graph shown, the following are the six internally vertex disjoint paths on the vertices $\{x,y,x',z\}$: $(x,y)$, $(x',a,x)$, $(x',y',b,y)$, $(z,a',x')$, $(z, \ldots, x)$ using thick edges, $(z, \ldots, z', \ldots, y)$ using thick edges.
\textbf{(b)} A $K_{4}$-subdivision on vertices $\{x,y,u'.v' \}$
}
\end{figure}

\begin{proof}[Lemma \ref{Lemma_helpingLemma}]
Assume that $e''$ is an internal edge in $G$. Then $e''$ is contained in at least two induced cycles $C_1$ and $C_2$ in $G$. Without loss of generality assume that $C_1$ is in $Enc(C_{max})$ and let $P = C_2 -e''$ be a path.
From Lemma \ref{Lemma_InducedCycleOutside}.(a), $V(C_2) \cap V(C_{max}) =\{x,y\}$. 
Thus the path between $x$ and $y$ in $T$ and the path $P$ are internally vertex disjoint.
As a consequence, there is an edge $(u,v)$ in $P$ but not in $T$; otherwise the tree $T$ has a cycle.
By Lemma \ref{Lemma_InducedCycleOutside}.(b), the fundamental cycle formed by the non-tree edge $(u,v)$ is of larger length than $C_{max}$ and also has $e$. Because $C_{max}$ is a maximum length fundamental cycle containing $e$, this is a contradiction. Therefore, $e''$ is an external edge in $G$.

We now prove the second part of the lemma.
Let $C'$ be the induced cycle containing $e$ in $Enc(C_{max})$. Suppose $C'$ and $C_{min}$ are different, then $e$ is being shared by two induced cycles. This contradicts that $e$ is an external edge. Thus $C'$ is $C_{min}$. Hence $C_{min} \subseteq Enc(C_{max})$. 
\qed 
\end{proof}

\subsection{Structural Properties of MASTs}

\noindent
A set $A$ of edges in $G$ is referred to as a \emph{safe} set for $G$, if $A$ is an iterative set of edges for $G$ and a minimum average stretch spanning tree of $G$ is in $G-A$.

\begin{theorem}
\label{Theorem_safeEdges}
Let $A$ be a safe set of edges for $G$ such that $\bound(A,G) \neq \emptyset$.
Let $e$ be an edge in $\bound(A,G)$ for which $\cost(e)$ is minimum. Then $A \cup \{e\}$ is a safe set for $G$.
\end{theorem}
\begin{proof}
For a safe set $A$, let $T^*$ be a minimum average stretch spanning tree of $G$; that is, $T^* \subset G-A$ as $\bound(A,G) \neq \emptyset$.
If $e \notin E(T^*)$, then we are done.
Assume that $e  \in E(T^*)$. Clearly, $A \cup \{e \}$ is an iterative set for $G$.
To show that $A \cup \{e \}$ is a safe set for $G$,
we use the technique of cut-and-paste to obtain a spanning tree $T'$ (by deleting the edge $e$ from $T^*$ and adding an appropriately chosen edge $e'$) and show that $\avgstr(T') \leq \avgstr(T^{*})$.\\

\noindent
Let $G_e$ be a 2-connected component in $G-A$ containing $e$ and $G_1, \ldots, G_k$ be the 2-connected components in $G-A$. For clarity, $G_e \in \{ G_1, \ldots, G_k\}$.
From Lemma \ref{Lemma_shareSingleEdgeNew}.(c), $G_e$ is a polygonal 2-tree. 
For $1 \leq i \leq k$, by Lemma \ref{Lemma_subGraphSpanningTree}.(a), $T_i = T^* \cap G_i$ is a spanning tree of $G_i$.
For the spanning tree $T^*$, let $C_{min}$ be the smallest fundamental cycle containing $e$ in $G_e$ and let $C_{max}$ be a largest fundamental cycle containing $e$  in $G_e$.
Let $e', e'' \in E(G_e)$ be the non-tree edges associated with $C_{min}$ and $C_{max}$, respectively. 
From Lemma \ref{Lemma_helpingLemma}, $e''$ is an external edge in $G_e$ and $Enc(C_{min}) \subseteq Enc(C_{max})$. Let $e'=(x_{min},y_{min})$, $e''=(x_{max},y_{max})$.
For a non-tree edge $(u,v)$ in $T^*$, we use $P_{uv}$ to denote the path between $u$ and $v$ in $T^*$ and $C_{uv}$ to denote the fundamental cycle of $T^*$ formed by $(u,v)$. 
Let 
 $X = \{ (u,v) \in E(G) \setminus E(T^*) \mid e\in E(P_{uv}), e' \notin Enc(C_{uv}) \}$, 
$Y=\{ (u,v) \in E(G) \setminus E(T^*) \mid e \in E(P_{uv}), e' \in Enc(C_{uv}), (u,v) \neq e' \}$,
 $Z = \{ (u,v) \in E(G) \setminus E(T^*) \mid e \notin E(P_{uv}) \}$. 
  The set of non-tree edges in $T^*$ is $X \uplus Y \uplus \{e'\} \uplus Z$. Let $T' = T^* + e' - e$.
The set of non-tree edges in $T'$ is $X \uplus Y \uplus Z \uplus \{ e \}$.  To prove the theorem, we prove the following claims. \\

\noindent
\textbf{Claim 1:} $X \subseteq A$.\\
\textbf{Claim 2:} $\Sup(e) \subseteq X$.\\
\textbf{Claim 3:} $\Sup(e'') \subseteq Y$.\\
\textbf{Claim 4:} $X \subseteq \Sup(e)$.\\
\textbf{Claim 5:} For every $(u,v) \in Z$, the path between $u$ and $v$ in $T^*$ is in $T'$.\\

\noindent
Assuming that the above five claims are true, we complete the proof of the theorem.  
We know that $\cost(e) \leq \cost(e'')$. As $e$ and $e''$ are in $G_e$, from the definition of $\Sup$, we further know that $\Sup(e) \cap \Sup(e'') = \emptyset$. Therefore, from Claims 2, 3 and 4, it follows that $|X| \leq |Y|$.
Since $e', e \in E(C_{min})$, $e \in E(T^*)$ and $e' \notin E(T^*)$, the stretch of $e'$ in $T^*$ is equal to the stretch of $e$ in $T'$. 
From Claim 5, stretch do not change for the edges in $Z$.
For all the edges in $X$, stretch increases by $|C_{min}|-2$. Further, for all the edges in $Y$, stretch decreases by $|C_{min}|-2$. 
If $|X|<|Y|$, as shown in Fig \ref{figNotequalcasethm}, then $\avgstr(T') < \avgstr(T^{*})$; it contradicts that $T^*$ is a minimum average stretch spanning tree. Thereby $|X| = |Y|$, shown in Fig \ref{figEqualcasethm}. This implies that $\avgstr(T') = \avgstr(T^{*})$.
Since $T^{*}$ is a minimum average stretch spanning tree, $T'$ is also a minimum average stretch spanning tree.
Clearly, $T'$ is in $G - (A \cup \{e\})$. 
Hence $A \cup \{e\}$ is a safe set for $G$.


\begin{figure}[!ht]
\centering
  \subfloat[$X= \{(g,c) \}, Y=\{(f,b)\}$ \label{figEqualcasethm}]{%
    \includegraphics[width=0.35\textwidth]{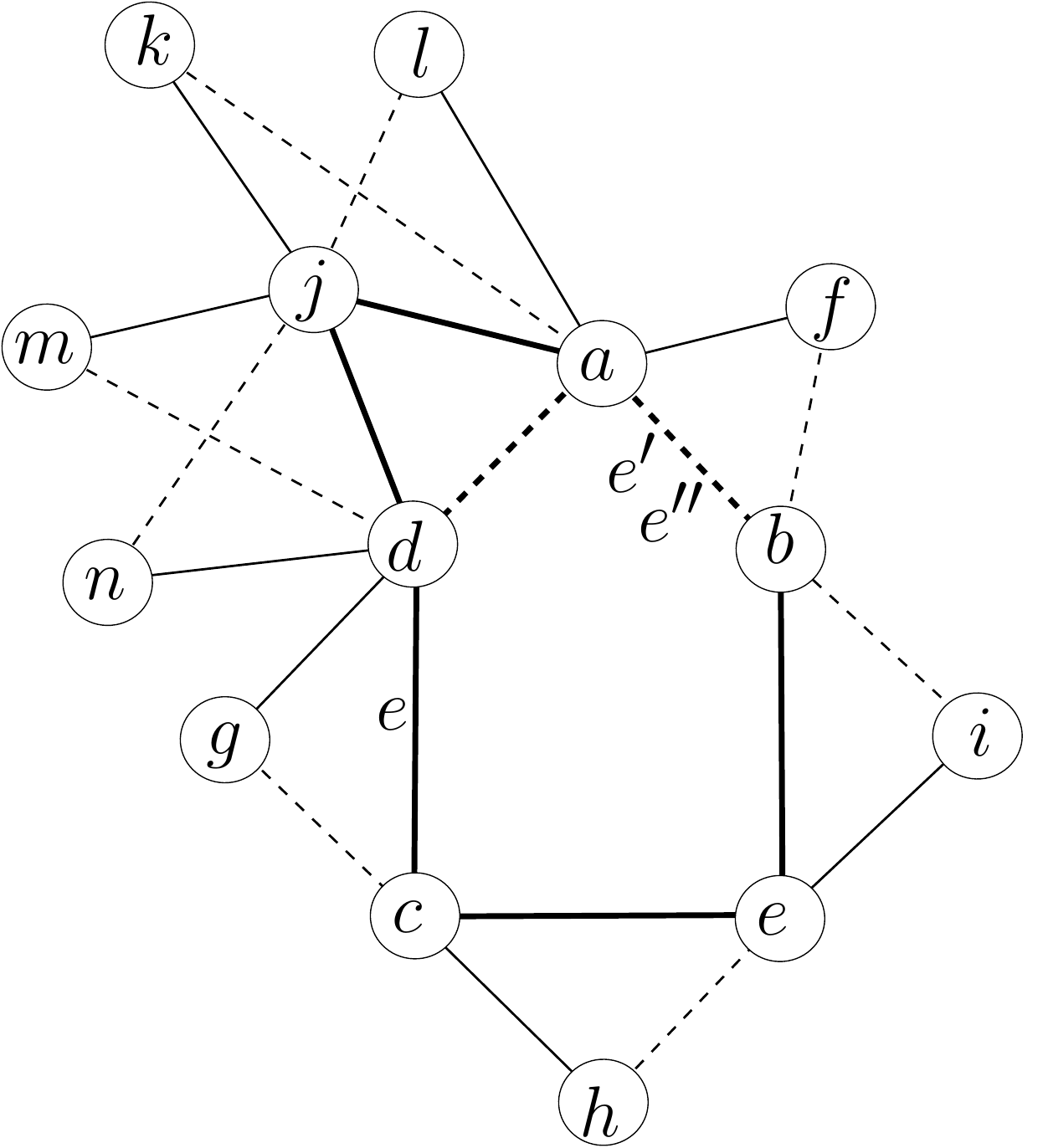}
  }
  \subfloat[$X= \{(g,c) \}, Y=\{(f,b),(f,p)\}$\label{figNotequalcasethm}]{%
    \includegraphics[width=0.35\textwidth]{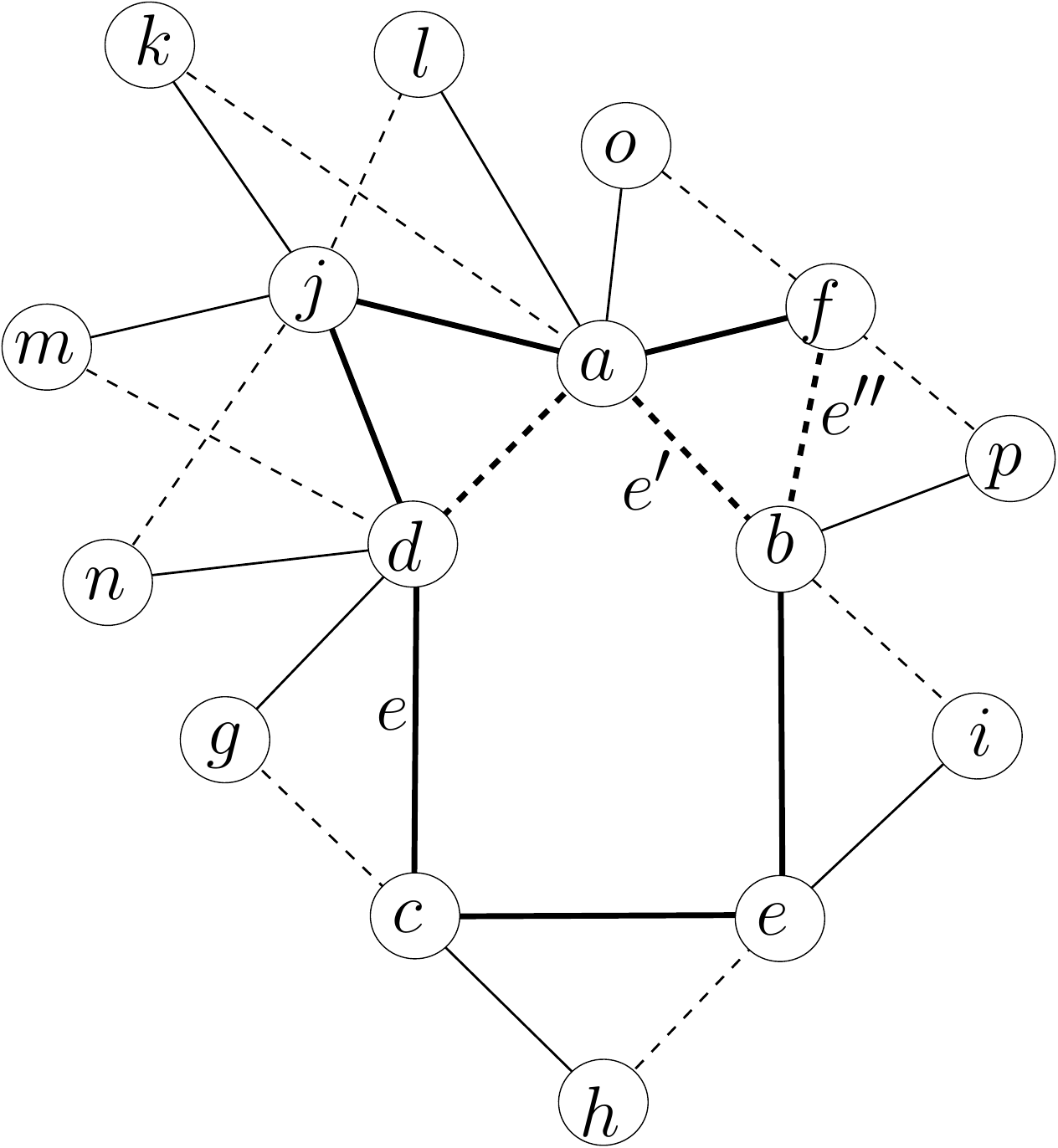}
  }
  \caption{Dashed and solid edges shown in thick are the edges of $G_e$. Dashed edges are the non-tree edges of $T^*$ and solid edges are the edges of $T^*$.}
\end{figure}

\noindent
We now prove the five claims.\\

\noindent
\textbf{Proof of Claim 1:}
On the contrary, assume that $(u,v) \in X$ and $(u,v) \notin A$. 
To arrive at a contradiction, we show that $e$ is an internal edge.
Since $(u,v) \in X$, there is a fundamental cycle $C_{uv}$ of $T^*$ formed by the non-tree edge $(u,v)$ containing $e$.
As $(u,v) \notin A$, clearly $(u,v)$ is in $G-A$. Further, $P_{uv}$ is in $G-A$, because  $T^* \subset G-A$.
So we know that $C_{uv}$ is in $G-A$. 
If $C_{uv}$ is not in $G_e$, then $G_e \cup C_{uv}$ becomes a 2-connected component in $G-A$, because $G_e$ is in $G-A$, $C_{uv}$ is in $G-A$, and $e$ is both in $G_e$ and $C_{uv}$.
 But, we know that $G_e$ is a maximal 2-connected subgraph (2-connected component), thereby $C_{uv}$ is in $G_e$.
Clearly, $C_{uv}$ and $C_{min}$ are not edge disjoint cycles.
If $Enc(C_{min}) \subseteq Enc(C_{uv})$, then either $(u,v) \in Y$ or $(u,v) = e'$, which contradicts the fact that $(u,v) \in X$.
Also, $Enc(C_{uv})$ is not contained in $Enc(C_{min})$, because $C_{min}$ is a minimum length induced cycle containing $e$.
Therefore, both $C_{min}$ and $C_{uv}$ are not contained in each other. Thus, $e$ is an internal edge in $G-A$. 
This is a contradiction, as we know that $e$ is external.\\

\noindent
\textbf{Proof of Claim 2:}
Let $(u,v) \in \Sup(e)$. 
In order to prove that $(u,v) \in X$, we show the following: (a) $(u,v) \notin E(T^*)$, (b) $P_{uv}$ has $e$ and (c)  $e'$ is not in $Enc(C_{uv})$.

By the definition of $\Sup(e)$, $(u,v) \in A$. As $T^* \subset G-A$, it follows that $(u,v) \notin E(T^*)$.
By Lemma \ref{LemmaSupportBothDirections}, there is a shortest path $P$ joining $u$ and $v$ in $G-A$ and $P$ has $e$. 
Let $G'_{1}, \ldots, G'_{r}$ be the 2-connected components in $G-A$ containing at least two vertices from $P$.
Due to Lemma \ref{Lemma_shareSingleEdgeNew}.(b), for each $1 \leq i \leq r$,  $P \cap G'_i$ is an edge, say $(x_i,y_i)$. 
Thus $P \cap G_e$ is $e$. Further, $P$ contains at most one vertex from $e'$, because $e' \in E(G_e)$.
The set of edges in $P$ that are cut-edges in $G-A$ are present in $T^*$. 
Due to Lemma \ref{Lemma_subGraphSpanningTree}.(a), replacing every edge $(x_i,y_i)$ in $P$ by the path between $x_i$ and $y_i$ in $T^*$, $P_{uv}$ is obtained.
Since $P \cap G_e$ is $e$ and $e$ is in $T^*$, it implies that $P_{uv}$ has $e$.  
Thus $P_{uv}$ has $e$, and $e'$ is not in $Enc(C_{uv})$.\\


\noindent
\textbf{Proof of Claim 3:}
Let $(u,v) \in \Sup(e'')$. 
In order to prove that $(u,v) \in Y$, we show the following: (a) $(u,v) \notin E(T^*)$, (b) $P_{uv}$ has $e$ and (c) $e'$ is in $Enc(C_{uv})$. 

Because $(u,v) \in A$ and $T^* \subset G-A$, we have $(u,v) \notin E(T^*)$.
As $e,e'' \in E(G_e)$, due to Lemma \ref{LemmaSupportBothDirections},
 there is a shortest path $P$ joining $u$ and $v$  in $G-A$ and $P$ has $e''$.
Let $G'_{1}, \ldots, G'_{r}$ be the 2-connected components in $G-A$ such that for each $1 \leq i \leq r$,  $P \cap G'_i$ is an edge, say $(x_i,y_i)$, due to Lemma \ref{Lemma_shareSingleEdgeNew}.(b).
By Lemma  \ref{Lemma_subGraphSpanningTree}.(a),
we replace every edge $(x_i,y_i)$ in $P$ by the path between $x_i$ and $y_i$ in $T^*$ and obtain the tree path $P_{uv}$.
Note that $P \cap G_{e}$ is $e''$, $e''=(x_{max},y_{max})$, and
 $e''$ in $P$ got replaced with the path between $x_{max}$ and $y_{max}$ in $T^*$.
Also, we know that the path between $x_{max}$ and $y_{max}$ in $T^*$ has $e$.
Further by Lemma \ref{Lemma_helpingLemma}.(b), $e'$ is in $Enc(C_{max})$. These observations imply that $P_{uv}$ has $e$ and $Enc(C_{uv})$ contains $e'$.\\

\noindent
\textbf{Proof of Claim 4:}
Let $(u,v) \in X$. By Claim 1, clearly $(u,v) \in A$.
Lemma \ref{Lemma_subGraphSpanningTree}.(b) implies that $P_{uv} \cap G_e$ is a path.
Let $P' = P_{uv} \cap G_e$ be a path and let $x$ and $y$ be the end vertices of $P'$.
If $P'$  is an edge, shown in Fig \ref{figLastcase1}, then the claim holds. 
On the contrary assume that $P'$ has at least two edges.
By Lemma \ref{Lemma_shareSingleEdgeNew}.(a), $(x,y) \in E(G)$. Further, $(x,y) \notin E(T^*)$ as it would then form a cycle in the tree. 
If $x_{min}, y_{min} \in V(P')$, shown in Fig \ref{figLastcase2} and Fig \ref{figLastcase3}, then $(u,v)$ must be in $Y$. As we know that $(u,v) \in X$,  the  path $P'$ is strictly contained in the path joining the vertices $x_{min}$ and $y_{min}$ in $T^*$.
Then the fundamental cycle of $T$ formed by $(x,y)$ is of lesser length than the length of $C_{min}$, shown in Fig \ref{figLastcase4}; a contradiction because $C_{min}$ is a minimum length fundamental cycle in $G_e$ containing $e$.
Therefore,  $P_{uv} \cap G_e$ is $e$. 
Thus $(u,v) \in \Sup(e)$.\\


\begin{figure}[!ht]
\centering
  \subfloat[ $x_{min}, y_{min} \in V(P')$ implies that $e' \in E(C_{uv})$ \label{figLastcase2}]{%
    \includegraphics[width=0.25\textwidth]{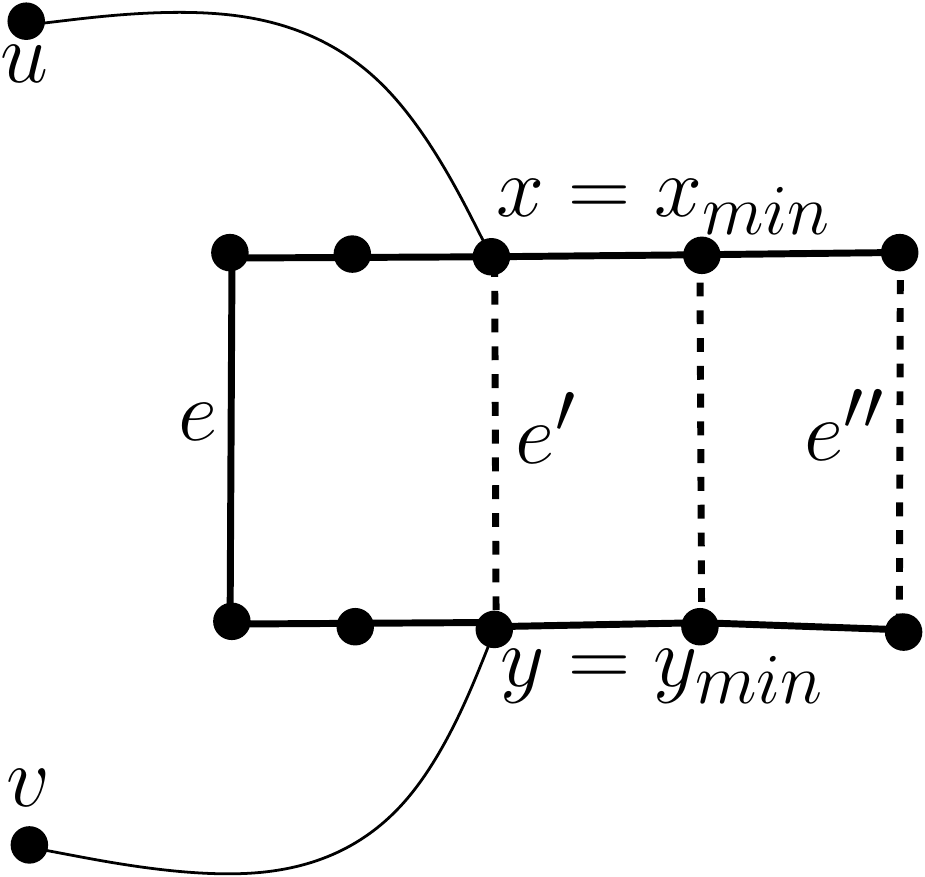}
  }
  \quad \quad
  \subfloat[$x_{min}, y_{min} \in V(P')$ implies that $e' \in E(C_{uv})$ \label{figLastcase3}]{%
    \includegraphics[width=0.25\textwidth]{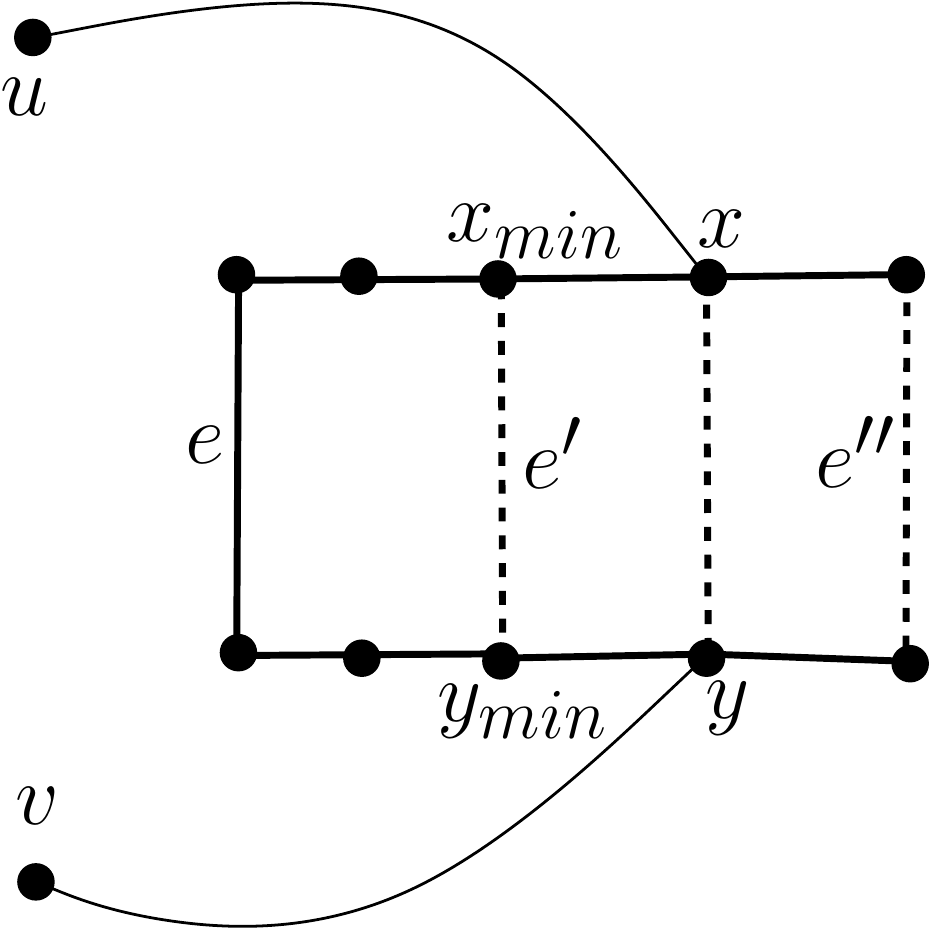}
  }

\quad \quad
  \subfloat[ $P' \cap G_e = e$\label{figLastcase1}]{%
    \includegraphics[width=0.25\textwidth]{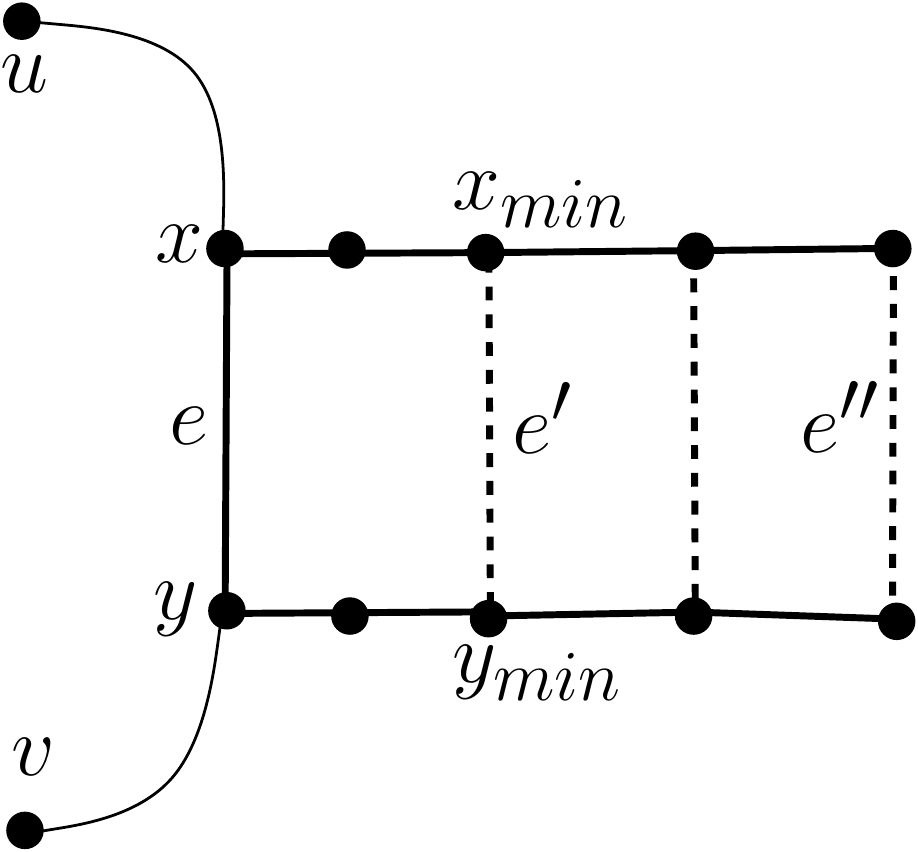}
  }
\quad \quad
  \subfloat[ $P'$ has at most one vertex from $\{x_{min}, y_{min} \}$ \label{figLastcase4}]{%
    \includegraphics[width=0.25\textwidth]{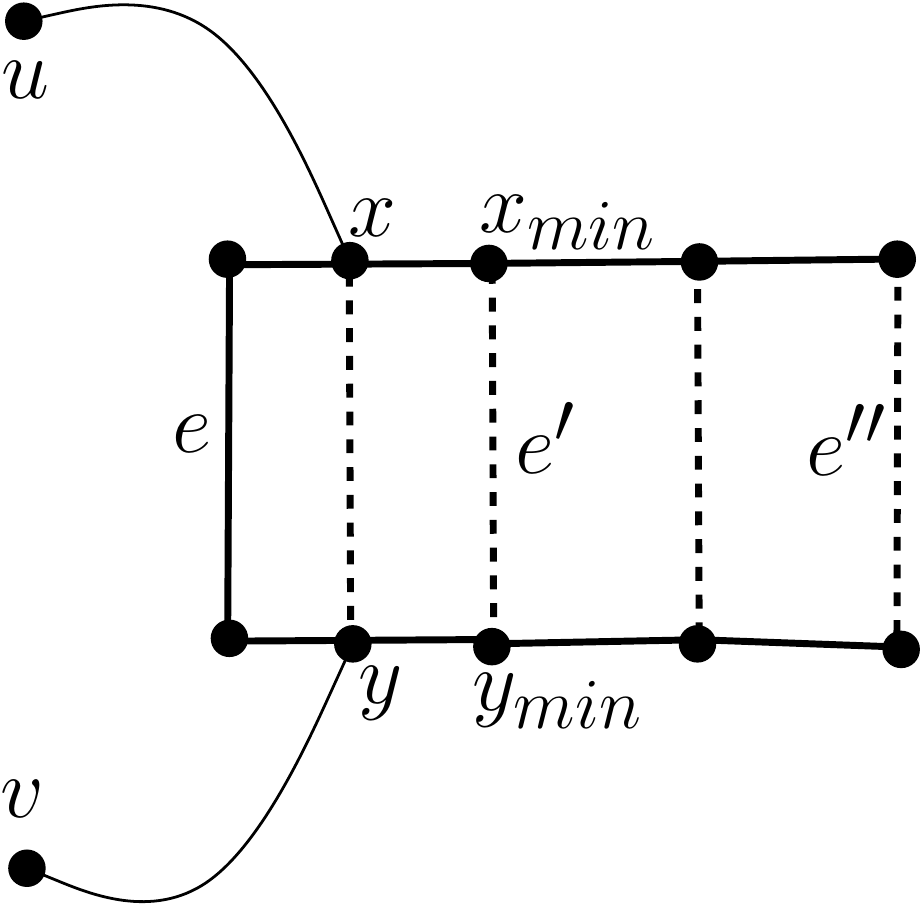}
  }

  \caption{Cases in Claim 4: Dashed edges and solid edges shown in thick are the edges of $G_e$. Solid edges shown in thick are the edges of $G_e \cap T^*$. The path $G_e \cap P_{uv}$ is denoted by $P'$.}
\end{figure}

\noindent
\textbf{Proof of Claim 5:}
Let $(u,v) \in Z$. By the definition of $Z$, clearly $e \notin E(P_{uv})$. It implies that $e' \notin Enc(C_{uv})$ as the path between the end vertices of $e'$ in $T^*$ has $e$. Therefore $P_{uv}$ has at most one end vertex from $e$ and $e'$. Since the symmetric difference of $E(T^*)$ and $E(T')$ is $\{e,e'\}$, the path $P_{uv}$ in $T^*$ remains same in $T'$.\\
Hence the theorem.
\qed
\end{proof}

\noindent We now show the termination condition for applying Theorem \ref{Theorem_safeEdges}.

\begin{lemma}
\label{Lemma_safeEdgesTermination}
Let $A$ be a safe set of edges for $G$ such that $\bound(A,G) = \emptyset$.
Then $G-A$ is a minimum average stretch spanning tree of $G$. 
\end{lemma}
\begin{proof}
 Since $A$ is a safe set for $G$, a minimum average stretch spanning tree is contained in $G-A$.
Since $\bound(A,G)$ is $\emptyset$, $G-A$ is acyclic. Therefore, $G-A$ is a minimum average stretch spanning tree of $G$.
\qed
\end{proof}

\section{Computation of an MAST in Polygonal 2-trees}
\label{Section_ComputeMAST}
In order to obtain a minimum average stretch spanning tree  efficiently, we need to efficiently find an edge in  $\bound(A,G)$ with minimum cost in every iteration, where $A$ is a safe set for $G$. In this section, we present necessary data-structures, so that a minimum average stretch spanning tree in polygonal 2-trees on $n$ vertices can be computed in $O(n \log n)$ time. 
A pseudo code for achieving this is given in Algorithm \ref{mastAlgo}.
For each edge $e \in \bound(A,G)$, we show how to compute $cost(e)$ efficiently in  Lemma \ref{LemmaCostUpdate}.

\noindent
\textbf{Notation.} Let $Q$ be a min-heap that supports the following operations: $Q$.\texttt{insert}($x$) inserts an arbitrary element $x$ into $Q$, $Q$.\texttt{extract-min}() extracts the minimum element from $Q$, $Q$.\texttt{decrease-key}($x,k$) decreases the key value of $x$ to $k$ in $Q$, $Q$.\texttt{delete}($x$)
deletes an arbitrary element $x$ from $Q$. $Q$.\texttt{delete}($x$) can be implemented by calling $Q$.\texttt{decrease-key}($x, - \infty$) followed by  $Q$.\texttt{extract-min}() \cite{CormenBook}.
For a set $A$ of safe edges for $G$, an induced cycle in $G$ is said to be \emph{processed} if it is not in $G-A$; otherwise it is said to be \emph{unprocessed}.
For an edge $e \in E(G)$,  we use the sets $Cycles[e]$ and $pCycles[e]$ to store the set of induced cycles and processed induced cycles, respectively containing $e$;  $unpCount[e]$ is used to store the number of unprocessed induced cycles containing $e$.
For an edge $e \in E(G)$, we use $c[e]$ to store some intermediate values while computing $\cost(e)$; whenever $e$ becomes an edge in $\bound(A,G)$, we make sure that  $c[e]$ is $\cost(e)$. \\

\noindent
\textbf{Initialization.} Given a polygonal 2-tree $G$, we first compute the set of induced cycles in $G$.
For each induced cycle $C$ in $G$ and for each edge $e \in E(C)$, we insert the cycle $C$ in the set $Cycles[e]$. For each $e\in E(G)$, we perform $unpCount[e] \leftarrow |Cycles[e]|$, $pCycles[e] \leftarrow \emptyset$, $c[e] \leftarrow 0$.
 We further initialize the set $A$ of safe edges with $\emptyset$.
Later, we construct a min-heap $Q$ with the edges $e$ in $\bound(A,G)$ i.e., external edges that are not bridges in $G$, based on $c[e]$.

\noindent Now we shall look at \textbf{Algorithm} \ref{mastAlgo}.
\begin{algorithm2e}
\caption{An algorithm to find an MAST of a polygonal 2-tree $G$}
\label{mastAlgo}

Perform the steps described in Initialization \;
\While{$Q \neq \emptyset$}
{ 
	 $e \leftarrow Q$.\texttt{extract-min}() \;
 	 $A \leftarrow A\cup \{e \}$ \;
 	 $C \leftarrow Cycles[e] \setminus pCycles[e]$ \label{LineExternalCycle}\; 
 		\For{each edge $\hat{e} \in E(C) \setminus \{e\}$}
		{
 			$c[\hat{e}] \leftarrow c[\hat{e}] + c[e]+ 1$ \label{LineUpdateCost}\;
			$pCycles[\hat{e}] \leftarrow pCycles[\hat{e}] \cup C$ \label{LineUpdateProcessed}\; 			
 			$unpCount[\hat{e}] \leftarrow unpCount[\hat{e}] - 1$ \label{LineUpdateUnProcessed}\;
			\textbf{if} {$unpCount[\hat{e}]=1$} \textbf{then}
				   $Q$.\texttt{insert}($\hat{e}, c[\hat{e}]$) \;
			
			\textbf{if} {$unpCount[\hat{e}]=0$} \textbf{then}			
				   $Q$.\texttt{delete}($\hat{e}$) \;
			
 		}       
 
 } 
Return $G - A$;
\end{algorithm2e}

\noindent
\textbf{Algorithm} \ref{mastAlgo} maintains the following loop invariants:

\begin{enumerate}
 \item[L1.] The min heap $Q$ only consists of, the set of edges in $\bound(A,G)$. 

\item[L2.] For an edge $e \in E(G)$, $pCycles[e]$  is the set of processed induced cycles containing $e$ and $unpCount[e]$ is equal to the number of unprocessed induced cycles containing $e$.
 
 \item[L3.] For every edge $e \in \bound(A,G)$, $Cycles[e] \setminus pCycles[e]$ is the unique external induced cycle in $G-A$ containing $e$.
 
\item[L4.] $A$ is a safe set for $G$. (cf. Theorem \ref{Theorem_safeEdges})
 \item[L5.] For every edge $e \in \bound(A,G)$, $c[e] = \cost(e)$. (cf. Lemma \ref{LemmaCostUpdate})
\end{enumerate}

\noindent
\textbf{Proofs of Loop Invariants L1-L3}
\begin{proof}
An edge gets inserted into $Q$ only when it is in a unique induced cycle of $G-A$. 
Further, all the bridges in $G-A$ are getting deleted in line 11. Thus $L1$ holds.
In lines 8 and 9, processed induced cycles and the count of unprocessed cycles are updated. Thus L2 holds.
For each edge $e\in \bound(A,G)$, $unpCount[e]$ is one. Therefore, $Cycles[e] \setminus pCycles[e]$ gives the unique induced cycle in $G-A$ containing $e$, thereby L3 holds.
\qed
\end{proof}

\noindent
The proof of loop invariant L5 is deferred to next subsection. We finish the running time analysis of {Algorithm} \ref{mastAlgo} below.
The algorithm terminates when $Q$ becomes $\emptyset$, that is, $\bound(A,G) = \emptyset$. Then by Lemma \ref{Lemma_safeEdgesTermination}, $G-A$ is a minimum average stretch spanning tree of $G$.

\begin{lemma}
 For a polygonal 2-tree $G$ on $n$ vertices, \textbf{Algorithm} \ref{mastAlgo} takes $O(n \log n)$ time.
\end{lemma}
\begin{proof}
The set of induced cycles in $G$ can be obtained in linear time (cf. Theorem \ref{TheoremPolygonal2TreeNiceEarDecomp}), thereby 
line 1 takes linear time.
As the size of induced cycles in $G$ is $O(n)$ (Theorem \ref{TheoremPolygonal2TreeNiceEarDecomp}),  line \ref{LineExternalCycle} and lines 7-9 contribute $O(n)$ towards the run time of the algorithm. 
Also every edge in $G$ gets inserted into the heap $Q$ and gets deleted from $Q$ only once and $|E(G)| \leq 2n-3$. 
  It takes $O( \log n)$ time for the operations \texttt{insert()},\texttt{delete()} and \texttt{extract-min()} \cite{CormenBook}. 
Thus \textbf{Algorithm} \ref{mastAlgo} takes $O(n \log n)$ time.
\qed 
\end{proof}

\subsection{Cost Updating Procedure}
During the execution of our algorithm, for each edge $e$ in $G-A$, such that $e$ is external and not a bridge, we need to compute $\cost(e)$ efficiently.
 This is done in \textbf{Algorithm} \ref{mastAlgo} in line \ref{LineUpdateCost}.
We prove the correctness of this step in Corollary \ref{Corollary_costUpdate} using Lemma \ref{LemmaCostUpdate}.\\

\noindent
The \textbf{Algorithm} \ref{mastAlgo} runs for $m-n+1$ iterations.
For $0 \leq j \leq m-n+1$, let $A_j \subset E(G)$ denote the set of safe edges in $G$ at the end of $j^{\text{th}}$ iteration. 
Let $e$ be an edge extracted from the heap $Q$ in  $j^{\text{th}}$ iteration and $C$ be the unique induced cycle containing $e$ in $G-A_{j-1}$.
That is, $C$ is a cycle in $G-A_{j-1}$ and $C$ is not a cycle in $G-A_{j}$ as $e$ is added to $A$ in iteration $j$.
Then we say that $C$ is processed in iteration $j$ and $e$ is the \emph{destructive} edge for $C$.
For each external edge in $G$, we know that $\Sup(e)$ is $\emptyset$, which implies that $\cost(e)$ is 0.

\begin{lemma}
\label{LemmaCostUpdate}
 Let $e \in \bound(A_j,G)$ such that $e$ is an internal edge in $G$, where $ 0 \leq j <  m-n+1$.
 Let $C$ be the unique external induced cycle in $G-A_{j}$ containing $e$ and 
 $C_{1}, \ldots, C_{k}$ be the other induced cycles in $G$ containing $e$. 
For $1 \leq i \leq k$, let  $e_{i}$ be the destructive edge of $C_i$. 
 Then $\Sup(e) = \Sup(e_1) \uplus \ldots \uplus \Sup(e_k) \uplus \{ e_1, \ldots, e_k\}$.
\end{lemma}
We use the following lemma to prove Lemma \ref{LemmaCostUpdate}. For a path $P$ and for vertices $x,y \in V(P)$, $P(x,y)$ denotes a subpath in $P$ with end vertices $x$ and $y$. 
For an edge $(x,y)$ in  $G$, if $(x,y)$ is external, then there is a unique shortest path between $x$ and $y$ in $G-(x,y)$.

\begin{lemma}
\label{LemmaHelpCostUpdate}
Let $A$ be a safe set for $G$ and $P$ be a path with end vertices $u$ and $v$ in $G-A$. 
Let $(a,b)$ be an edge in $P$ such that $(a,b) \in \bound(A,G)$.
 Let $G' = G-(A \cup \{(a,b)\})$ and 
$P'$ be the path obtained from $P$ by replacing $(a,b)$ with the shortest path between $a$ and $b$ in $G'$.
Then $P$ is a shortest path in $G-A$ if and only if $P'$ is a shortest path in $G'$.
\end{lemma}

\begin{proof}
($\Rightarrow$)
Assume that $P'$ is not a shortest path joining $u$ and $v$ in $G'$.
Then there exist a path $P''$ joining $u$ and $v$ in $G'$ such that $|E(P'')| < |E(P')|$.
Consider the case when $P''$ has both $a$ and $b$. 
Because there is a unique shortest path joining $a$ and $b$ in $G'$, $P'(a,b)$ is $P''(a,b)$.
Then, without loss of generality the path $P''(u,a)$ consisting of lesser number of edges than $P'(u,a)$. As $P(u,a)$ is $P'(u,a)$, replacing the path $P(u,a)$ in $P$ with $P''(u,a)$ leads to a path shorter than $P$ in $G-A$, which contradicts that $P$ is a shortest path.
Consider the other case when $P''$ has at most on vertex from $\{a,b\}$.
In the path $P$ from $u$ to $v$, without loss of generality, assume that $a$ appears before $b$. 
In the sequence of vertices in $P$ from $u$ to $a$, let $a'$ be the last vertex in $P''$.
Similarly in the sequence of vertices in $P$ from $b$ to $v$, let $b'$ be the first vertex in $P''$.
Let $G''$ be the graph obtained by performing union on  $P(u,v) \cup \{(u,v)\}$ and the polygonal 2-tree containing $(a,b)$ in $G-A$. Since the intersection of $P(u,v) \cup \{(u,v)\}$ and $G-A$ is $(a,b)$, $G''$ is a polygonal 2-tree.
Now we have three internally vertex disjoint paths between $a'$ and $b'$ in $G''$: $P''(a',b')$, $P(a',b')$, and the path other than $P(a',b')$ in $P(u,v) \cup \{(u,v)\}$. Since $|\{a,b\} \cap \{a',b'\}| \leq 1$ and $P$ is a shortest path in $G-A$, $(a',b') \notin E(G'')$. By Lemma \ref{LemmaAtLeastThreeComponents}, $G'' -\{a',b'\}$ has at least three components. Applying
 Lemma \ref{Lemma_Polygonal2TreeProperty}(c), $G''$ is not a polygonal 2-tree, which contradicts that $G''$ is a polygonal 2-tree.
\\

($\Leftarrow$) Assume that $P$ is not a shortest path joining $u$ and $v$ in $G-A$.
Then there exist a path $P''$ joining $u$ and $v$ in $G-A$ such that $|E(P'')| < |E(P)|$.
Consider the case where $a,b \in V(P'')$. $P''$ is a disjoint union of $P''(u,a)$, $P''(a,b)$ and $P''(b,v)$. Similarly $P$ is a disjoint union of $P(u,a)$, $P(a,b)$ and $P(b,v)$. Since  $P(a,b)=P''(a,b)=e$, without loss of generality the path $P''(u,a)$ consisting of lesser number of edges than $P(u,a)$. Replacing the path $P'(u,a)$ in $P'$ with $P''(u,a)$ leads to a path shorter than $P'$ in $G'$, which contradicts that $P'$ is a shortest path.
Consider the other case where at most one vertex from $\{a,b\}$ is in $P''$. 
Then $P''$ is in $G'$, and also we have $|E(P'')| < |E(P)| < |E(P')|$.
Consequently, $P''$ is  shorter than $P'$ in $G'$. This contradicts that $P'$ is a shortest path joining $u$ and $v$ in $G'$.
\qed
\end{proof}


\begin{proof}[of Lemma \ref{LemmaCostUpdate}]
For $1 \leq i \leq k$, let  $f(i)+1$ be the iteration number in which $C_i$ is processed in Algorithm \ref{mastAlgo}..

($\Leftarrow$)
Let $(u,v) \in \Sup(e_i)$ for some $1 \leq i \leq k$.
Then by Lemma \ref{LemmaSupportBothDirections},  there is a shortest path $P$ joining $u$ and $v$ in $G-A_{f(i)}$ and $P$ has $e_i$.
From the premise, $e_i$ gets added to $A$ in the iteration $f(i)+1$.
Thereby $e_i \in A_{f(i)+1}$, which implies that $e_i$ is not in $G-A_{f(i)+1}$.
Consider the path $P' = (P - e_i) \cup (C_i - e_i)$ in $G-A_{f(i)+1}$.
As $e_i$ is exterior in $G-A_{f(i)+1}$, $C_i-e_i$ is a shortest path between the end vertices of $e_i$ in $G-A_{f(i)+1}$.
Also, $C_i-e_i$ has $e$.
Thus $P'$ has $e$.
By forward direction of Lemma \ref{LemmaHelpCostUpdate}, $P'$ is a shortest path joining $u$ and $v$ in $G-A_{f(i)+1}$.
 Note that $e$ is in $G- A_j$.
 By forward direction of Lemma \ref{LemmaHelpCostUpdate}, it follows that there is a shortest path $P_j$ joining $u$ and $v$ in $G-A_j$ and  $P_j$ has $e$.
Thus $(u,v) \in \Sup(e)$. 
Observe that $e_i \in Enc(P_{j} \cup (u,v))$.  
Further, the path between the end vertices of $e_i$ in $P_j$ is a shortest path containing $e$ in $G-A_j$.
Therefore, we have $e_i \in \Sup(e)$.

($\Rightarrow$)
Let $(u,v) \in \Sup(e)$. Then by Lemma \ref{LemmaSupportBothDirections}, there is a shortest path $P'$ joining $u$ and $v$ in $G - A_j$ and $P'$ has $e$.
Now $Enc(P' \cup (u,v))$ contains $e_i$ for some $1 \leq i \leq k$.
Consider the case $(u,v) \neq e_i$.
Let $x_i$ and $y_i$ be the end vertices of $e_i$.
We replace the path between $x_i$ and $y_i$  in $P'$ by the edge $e_i$ and let $P$ be the resultant path consisting of $e_i$.
From the reverse direction of Lemma \ref{LemmaHelpCostUpdate}, $P$ is a shortest path joining $u$ and $v$ in $G-A_{f(i)}$, and $P$ has $e_i$.
Thereby $(u,v) \in \Sup(e_i)$.
  As a result, $(u,v) \in \Sup(e_i) \cup \{ e_i \}$ for some $1 \leq i \leq k$.
\qed
\end{proof}

\noindent
Lemma \ref{LemmaCostUpdate} is illustrated in Fig \ref{figSupporLemma}.

\begin{figure*}[htp!]
\centering
\includegraphics[scale=0.45]{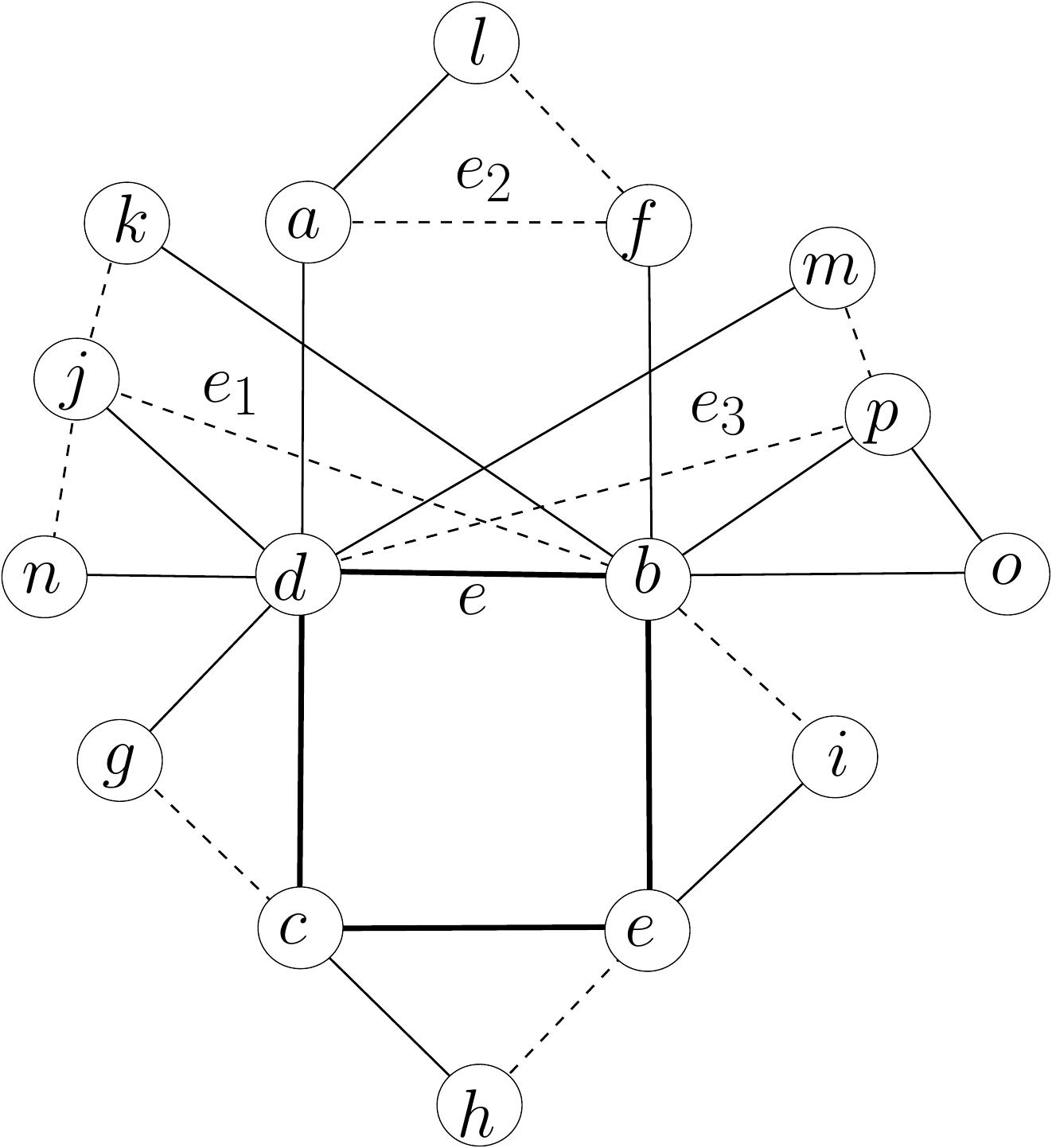}
\caption{For the polygonal 2-tree $G$ shown, dashed edges are the edges in $A$, solid edges are the edges of $G-A$ and solid edges shown in thick are the edges of $G_e$.  $\Sup(e_1) = \{(j,k)\}$, $\Sup(e_2) = \{(l,f)\}$, $\Sup(e_3)=\{(m,p)\}$ and $\Sup(e) = \{(j,k),(l,f),(m,p),e_1,e_2,e_3\}$}
\label{figSupporLemma}
\end{figure*}

\begin{corollary}
\label{Corollary_costUpdate}
 Let $e, e_1, \ldots, e_k$ be the edges as mentioned in Lemma \ref{LemmaCostUpdate}. Then $\cost(e) = \cost(e_1) + \ldots + \cost(e_k)$.
\end{corollary}


\noindent
This concludes the presentation of our main result, namely Theorem \ref{Theorem_mainResult}. 

\section{On Minimum Cycle Bases in Polygonal 2-trees}
\label{SectionCyclebasis}
This section is of our independent interest, which presents results on minimum cycle basis in polygonal 2-trees.
In particular, we show that there is a unique minimum cycle basis in polygonal 2-trees, which can be computed in linear time.
Also, we present an alternative characterization for polygonal 2-trees using cycle basis. This is shown in Theorem \ref{TheoremCharacterizationPolygonal2Trees}.

A graph is  \emph{Eulerian}  if the degree of every vertex is even.
Let G be an arbitrary unweighted graph and $H_1,\ldots, H_k$ be subgraphs of $G$.
Then the graph $H_1 \oplus \ldots \oplus H_k$ consists only the edges that appear odd number of times in $H_1, \ldots, H_k$. 
A minimal set $\mathcal{B}$ of Eulerian subgraphs of $G$ is a \emph{cycle basis} of $G$, if every cycle in $G$ can be expressed as exclusive-or ($\oplus$) sum of a subset of graphs in $\mathcal{B}$.
A \emph{minimum cycle basis} (MCB) of $G$ is a cycle basis that minimizes the sum of the lengths of the cycles in the  cycle basis. It is well known that every Eulerian graph in any minimum cycle basis is a cycle.
The cardinality of a cycle basis is $m-n+1$ \cite{Diestel2010}.

Planar graphs and Halin graphs are characterized based on their cycle bases.
A cycle basis is said to be \emph{planar basis} if every edge in the graph appears in at most two cycles in the cycle basis.
A graph is planar if and only if it has a planar basis \cite{MacLane37}.
A 3-connected planar graph is Halin if and only if it has a planar basis and every cycle in the planar basis has an external edge \cite{Syslo}.
Halin graphs that are not necklaces have a unique minimum cycle basis \cite{StadlerHalinCB}.
Also, outerplanar graphs have a unique minimum cycle basis \cite{LeydoldOuterPlanarChar}.

\begin{lemma}
 (Proposition 1.9.1 in \cite{Diestel2010}) \label{Lemma_inducedCyclesDiestel}
The induced cycles in an arbitrary graph $G$ generate its entire cycle space.
\end{lemma}

\begin{lemma}
\label{Lemma_inducedCyclesCount} The number of induced cycles in a polygonal 2-tree $G$ is $m-n+1$.
\end{lemma}
\begin{proof}
We apply induction on the number of internal edges in $G$.
Let $E_{in}(G)$ denote the set of internal edges in $G$.
If $|E_{in}(G)|=0$, then $G$ has one induced cycle and $m-n+1$ is one.
For the induction step, let $|E_{in}(G)| >0$.
We decompose $G$ into polygonal 2-trees $G_1, \ldots, G_k$ such that $G_1 \cup \ldots \cup G_k = G$ and $G_1 \cap \ldots \cap G_k$ is an edge in $G$, where $k\geq2$. Let $m_i = |E(G_i)|$ and $n_i = |V(G_i)|$. For every $1 \leq i \leq k$, $|E_{in}(G_i)| < |E_{in}(G)|$ as one internal edge of $G$ has become external in $G_i$. By induction hypothesis, for every $1 \leq i \leq k$, the number of induced cycles in $G_i$ is $m_i - n_i +1$.
Observe that the set of induced cycles in $G$ is equal to the disjoint union of the set of induced cycles in $G_1 , \ldots, G_k$.
 Further, we know that $m = m_1 + \ldots + m_k -k +1$ and $n = n_1 + \ldots + n_k - 2k + 2$. 
 Consequently, we can see that the number of induced cycles in $G$ is $m-n+1$.
\qed
\end{proof}

\begin{lemma}
\label{Lemma_PolygonalInducedCycles}
 For an arbitrary 2-connected partial 2-tree $G$, if the set of induced cycles in $G$ is a cycle basis, then $G$ is a polygonal 2-tree.
\end{lemma}
\begin{proof}
Assume that $G$ is not a polygonal 2-tree. Then by Lemma \ref{Lemma_Polygonal2TreeProperty}, there exist two induced cycles $C_1$ and $C_2$ in $G$ such that $|E(C_1) \cap E(C_2)| \geq 2$. Let $C_3 =  C_1 \oplus C_2$. 
Since $C_1$ and $C_2$ are induced cycles, clearly $C_1 \cap C_2$ is a path and $C_3$ is a cycle.
Let $P$ be the maximal common path in $C_1$ and $C_2$. Let $P_1$ and $P_2$ be the maximal private paths in $C_1$ and $C_2$, respectively.

Consider the case when $C_3$ is an induced cycle. The set $\{ C_1, C_2, C_3 \}$ do not be a part of a cycle basis. It contradicts that the set of induced cycles in $G$ is a cycle basis. 

Consider the other case when $C_3$ is not an induced cycle. Let $C'_1, \ldots, C'_k$ be the set of induced cycles in $Enc(C_3)$.
Since $C_1$ and $C_2$ are induced cycles and $C_3$ is not an induced cycle, there exist a chord $e$ in $C_3$ such that, one end vertex of $e$ is in $P_1$ and the other end vertex of $e$ is in $P_2$.
Note that at least one induced cycle in $Enc(C)$ has $e$, where as $C_1$ and $C_2$ do not have $e$.
It follows that  $\{C_1, C_2 \}$ is different from $ \{C'_1, \ldots, C'_k \}$.
We can express $C_3$  as $C_1 \oplus C_2$  and as well as $C'_1 \oplus \ldots \oplus C'_k$. Therefore, $\{C_1, C_2 \} \cup \{C'_1, \ldots, C'_k \}$ do not be part of a cycle basis. It is a contradiction, because we know that the set of induced cycles in $G$ is a cycle basis.
Therefore, our assumption is incorrect and hence $G$ is a polygonal 2-tree.
\qed
\end{proof}

\begin{theorem}
\label{Theorem_ICyclesuniqueMCB}
For a polygonal 2-tree $G$, the set of induced cycles is a unique minimum cycle basis.
\end{theorem}
\begin{proof}
Recall that the cardinality of a cycle basis is $m-n+1$.
Therefore, from Lemma \ref{Lemma_inducedCyclesDiestel} and Lemma \ref{Lemma_inducedCyclesCount}, it follows that induced cycles in $G$ is a cycle basis.
Assume that $\mathcal{B}$ is a minimum cycle basis of $G$ such that $\mathcal{B}$ contains at least one non-induced cycle. 
 Let $C$ be a smallest non-induced cycle in $\mathcal{B}$ and $C_1, \ldots, C_k$ be the set of induced cycles in $Enc(C)$. Observe that $C_1 \oplus \ldots \oplus C_k$ is $C$. Clearly, there exists $1 \leq i \leq k$ such that $C_i  \notin \mathcal{B}$ as $\mathcal{B}$ is a cycle basis. We replace $C$ with $C_i$ and obtain a cycle basis such that its size is strictly less than the size of $\mathcal{B}$ as $|E(C)| > |E(C_i)|$. We have got a contradiction, because $\mathcal{B}$ is a minimum cycle basis. Therefore, the set of induced cycles in $G$ is a unique minimum cycle basis of $G$.
\qed
\end{proof}

The following theorem follows from Lemma \ref{Lemma_PolygonalInducedCycles} and Theorem \ref{Theorem_ICyclesuniqueMCB}.

\begin{theorem}
\label{TheoremCharacterizationPolygonal2Trees}
A graph $G$ is a polygonal 2-tree if and only if $G$ is a 2-connected partial 2-tree and the set of induced cycles in $G$ is a cycle basis. 
\end{theorem}

As the set of induced cycles in polygonal 2-trees is a minimum cycle basis, Theorem \ref{TheoremPolygonal2TreeNiceEarDecomp} computes a minimum cycle basis in polygonal 2-trees in linear time.



} 

\noindent
\textbf{Concluding Remarks.}
For a polygonal 2-tree on $n$ vertices, 
we have designed an $O(n \log n)$-time algorithm for the problem, \textsc{Mast}, of finding a minimum average stretch spanning tree.
By using this algorithm, we  have obtained a minimum fundamental cycle basis $\mathcal{B}$ of a polygonal 2-tree on $n$ vertices in $O( n \log n)$ + $\size(\mathcal{B})$ time. We have also shown that polygonal 2-trees have a unique minimum cycle basis and it can be computed in linear time. 
The problem of finding a minimum routing cost spanning tree is closely related to \textsc{Mast}. 
A \emph{minimum routing cost spanning tree} is a spanning tree of a graph that minimizes the sum-total distance between every two vertices in the spanning tree. 
The complexity of finding a minimum routing cost spanning tree in polygonal 2-trees (also in planar graphs) is open, where as it is NP-hard in weighted undirected graphs.

\bibliographystyle{abbrv}
\bibliography{references}

\end{document}